\documentclass{article}
\usepackage{xcolor}
\usepackage{natbib}
\usepackage{authblk}
\usepackage[margin=1in]{geometry}
\usepackage[colorlinks=true, linkcolor=blue, citecolor=blue, urlcolor=blue]{hyperref}
\usepackage{graphicx} 
\usepackage{amsthm}
\RequirePackage{amsthm,amsmath,amsfonts,amssymb,graphicx,mathtools,enumerate,placeins}
\newtheorem{theorem}{Theorem}
\newtheorem{proposition}[theorem]{Proposition}
\newtheorem{lemma}[theorem]{Lemma}
\newcommand{\Fdr}{\textnormal{Fdr}} 
\newcommand{\lfdr}{\textnormal{lfdr}}

\newcommand{\q}{\alpha}
\newcommand{\argmax}{\textnormal{argmax}}
\newcommand{\E}{\mathbb{E}}
\renewcommand{\P}{\mathbb{P}}

\newcommand{\FDP}{\textnormal{FDP}}

\newcommand{\clar}{\textnormal{clar}}

\newcommand{\iid}{\textnormal{iid}}

\newcommand{\R}{\mathbb{R}}

\newcommand{\de}{\textnormal{d}}
\newcommand{\sgn}{\textnormal{sgn}}

\newcommand{\bX}{\textbf{X}}
\newcommand{\idx}{j}

\usepackage[shortlabels]{enumitem}

\newcommand{\calN}{\mathcal{N}}

\newtheorem{remark}{Remark}[section]

\newtheorem{assumption}{Assumption}[section]

\title{Estimating the local false discovery rate \\ under an unknown symmetric null}
\author[1]{Daniel Xiang}
\author[2]{William Fithian}
\author[1]{Nikolaos Ignatiadis}
\author[3]{Jake A. Soloff}
\author[4]{Asaf Weinstein}
\affil[1]{\textit{University of Chicago}}
\affil[2]{\textit{University of California, Berkeley}}
\affil[3]{\textit{University of Michigan}}
\affil[4]{\textit{Hebrew University of Jerusalem}}
\date{\today}

\usepackage[normalem]{ulem}

\begin{document}

\maketitle

\begin{abstract}
This paper is concerned with estimating the local false discovery rate (lfdr) in a two-groups model where the only assumption regarding the null distribution is symmetry about zero. Our motivation comes from the contemporary framework for multiple hypothesis testing, particularly relevant in variable selection problems, which transforms any user-specified scores into statistics whose null distributions are 
symmetric about zero, whereas enrichment to the right of zero is generally expected for the non-nulls. While modern methods such as the knockoff filter \citep{barber2015controlling} are able to exploit the null property for controlling the false discovery rate (FDR), an arguably more appropriate goal is to target control of the {\em local} false discovery rate for the rejected hypotheses, as proposed in \citet{soloff2024edge} where the standard two-groups model (known $f_0$ and independence) is analyzed. 
Here, we take a step in this direction and propose to estimate the lfdr by targeting the surrogate density ratio $f(-w)/f(w)$, for $w>0$, where $f$ is the marginal density in the aforementioned ``stripped-down'' two-groups model. 
We study several estimators and propose a logistic regression based method with natural cubic spline basis.
We also show that any consistent estimator of this surrogate yields asymptotic lfdr control of the multiple testing procedure that thresholds the estimate at the nominal level. 
\end{abstract}

\section{Introduction}
\label{sec:intro}

In multiple linear regression with many predictor variables, a primary question of interest is to identify which ones are actually relevant to the conditional distribution of the response. 
In the {\em controlled} variable selection framework, the problem is formulated as simultaneously testing the null hypotheses
\begin{align}
\label{null-hypotheses}
    H_{0\idx}: \beta_\idx = 0, \hspace{2em} \idx=1,\dots,p, 
\end{align}
where $\beta_\idx$ is the effect of predictor $\idx$ on the mean response, holding all other predictors fixed. Selection of the $\idx^{\text{th}}$ variable then corresponds to rejecting the null hypothesis $H_{0\idx}$, and selection of
an irrelevant variable incurs a type I error. 
This framing appears to invite 
drawing upon the rich and long-standing literature on multiple hypothesis testing, however the traditional paradigm has substantial limitations in the setting above: first, the test statistics are usually not independent, and second, the null distribution is generally unknown when moving beyond plain least squares, which precludes computation of $p$-values.

These limitations have led \cite{barber2015controlling} to propose a strikingly different approach to the controlled 
variable selection problem. 
Instead of relying on $p$-values, their method constructs for each predictor variable a fake control, called a \textit{knockoff}, such that under the null, the knockoff variable's importance statistic (a measure of each predictor's relevance to the response) is exchangeable with that of the genuine variable. 
Comparing each importance statistic to its knockoff counterpart produces a `contrast statistic', known as a {\em $W$-statistic}, that is symmetrically distributed around zero under the null (in fact, the signs are not only $\pm1$ with equal probabilities, but also independent given the absolute values $|W_1|,\dots,|W_p|$).
Ultimately, this null property of the $W$-statistics is exploited to design a procedure that controls the {\em false discovery rate}  \citep[FDR,][]{benjamini1995controlling} for the selected subset of predictors.
\begin{figure}[t]
    \centering
    \includegraphics[width=\linewidth]{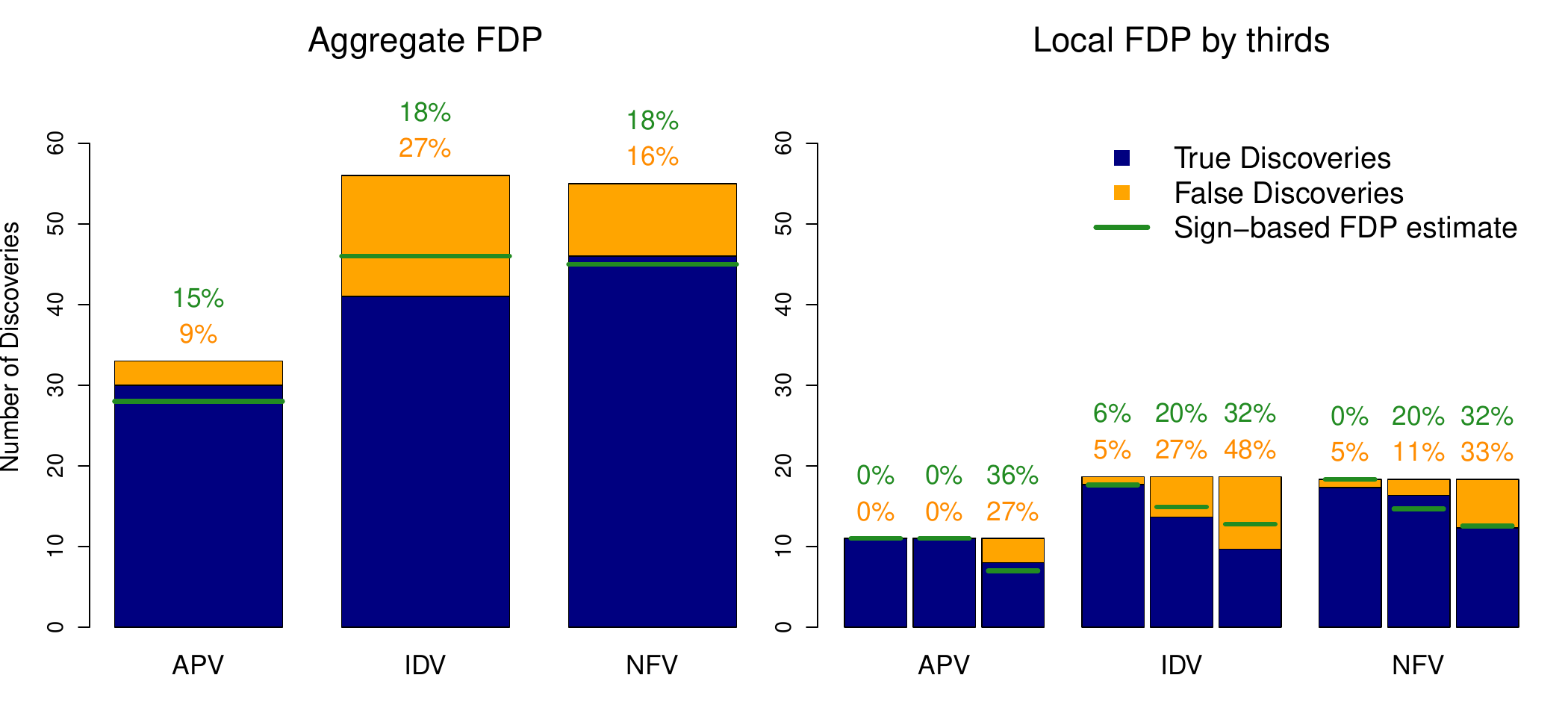}
    \caption{HIV drug discovery example. The left panel shows the Knockoffs rejection set targeting 20\% FDR for each of three antiviral drugs: Amprenavir (APV), Indinavir (IDV), and Nelfinavir (NFV). 
    The right panel shows each discovery set split into thirds, ordered from largest to smallest $W$-statistic. 
    Orange percentages show a validation estimate of the FDP based on an independent data set \citep[TSM,][]{rhee2005hiv}, while green percentages show a rough 
    sign-based estimate $\#\{W_\idx<0\}/\# \{W_\idx > 0\}$ in each bin. The TSM-based FDP increases across thirds from strongest to weakest, and the sign-based estimates roughly track this heterogeneity.}
    \label{fig:HIV-example-simple}
\end{figure}

The introduction of knockoffs has given rise to a large body of subsequent work on controlled variable selection methods, including the follow-up model-X knockoffs framework \citep{candes2018panning} for random-X settings with known covariate distribution, as well as the Gaussian Mirrors method of \citet{xing2023controlling} and their extensions \citep{dai2023falsea, dai2023scalefree, xia2023adaptive, wang2025adaptive}. 
These methods all rely on generating artificial null counterparts rather than deriving the null distribution of importance statistics. 
As noted in \cite{arias2017distribution}, 
this motivates studying a general model that only assumes symmetry under the null, which we take as the starting point for the current article.
To be useful in practice, the aforementioned methods additionally require the observations to be enriched on one side of zero under the alternative hypothesis, a condition that arises naturally in the regression setting for contrast statistics.\footnote{We remark that this model arises readily in other contexts, for example as a relaxation of one-way analysis of variance  models that usually assume a {\em particular} symmetric distribution for the noise, normality being the usual example \citep[see, e.g.,][]{tian2025conformalized}.}

A main advantage of $p$-value-free methods such as knockoffs is their flexibility in accommodating a wide range of importance statistics and design matrices while maintaining rigorous type I error control. 
However, there are legitimate concerns regarding the specific type I error metric that knockoffs are designed to control.  
While the FDR has been widely adopted
in modern multiple testing, \citet{soloff2024edge} rethink the criterion and explain why it fails to protect against type I errors at the level of individual discoveries.
In a two-groups model, the Bayesian (tail) FDR is equal to the {\em mean} of the \textit{local false discovery rate} (lfdr) conditionally on selection \citep[][see also \eqref{eq:sec-tan-fdr}]{efron2008microarrays}, so rejections with lfdr considerably higher than $\alpha$ are still possible when controlling the FDR at $\alpha$. 
In other words, the FDR criterion allows high-quality rejections to compensate for low-quality ones, compromising the reliability of discoveries made near the threshold.
This heterogeneity in the quality of discoveries is illustrated in Figure \ref{fig:HIV-example-simple} using the HIV drug resistance data example from \cite{rhee2005hiv}. 
After splitting the rejection set into three equally sized strata according to the size of the knockoff statistic $W_\idx$, we compare the FDP of each stratum, both as estimated using (i) a validation estimate based on an independent data set, and (ii) an analog of the knockoff filter's sign-based estimator, where we calculate $\#\{W_\idx<0\}/\# \{W_\idx > 0\}$ only for the variables in that stratum. As we see, the FDP varies systematically across these strata, with the marginal rejections (rejections with $W_\idx$ just above the rejection threshold $\hat{w}$) having a strikingly higher FDP than the rejections whose $W_\idx$ values are comfortably above the threshold.

To mitigate this deficiency,  \cite{soloff2024edge}~propose instead to control the \textit{max-lfdr}, defined as the expectation of the maximum lfdr among all rejected hypotheses. Controlling the max-lfdr at $\alpha$ ensures that, in expectation, the posterior probability of {\em each} rejection does not exceed $\alpha$. 
In addition to this new criterion, they propose a procedure (``support line") 
which operates on $p$-values and controls the max-lfdr in finite sample, assuming  an independent two-groups model where the null distribution is {\em known} and the alternative density is non-increasing.

Adopting this viewpoint, 
our impetus in the current paper is to design methods that target the lfdr, instead of the FDR, but move beyond the classical independent $p$-values setup to the flexible controlled variable selection setup. 
This means that we relax the assumption of a fully known null distribution to merely symmetry about zero. 
We will use the notation $W_\idx$ to denote the observed test statistics, which is standard notation in variable selection with knockoffs,\footnote{This is just one particular example of how the model \eqref{eq:two-group} could arise; see Section \ref{sec:applications} for an application with natively symmetric null statistics.} and will work under an {\em empirical Bayes} (EB) two-groups model, assuming that for $\idx =1,...,p$, the pairs $(W_\idx, H_\idx)$  are identically, but not necessarily independently,  distributed as
\begin{equation}
\label{eq:two-group}
    \begin{aligned}
    &\P(H_\idx = 0) = \pi_0,  \quad \quad   &&\P(H_\idx = 1) = 1-
    \pi_0\\[3pt]
    &W_\idx \mid H_\idx = 0\ \sim f_0,   \quad \quad &&W_\idx \mid H_\idx = 1 \ \sim f_1,
\end{aligned}
\end{equation}
with $\pi_0$ and $f_1$ unknown, and $f_0$  symmetric about zero but otherwise unknown.

The max-lfdr criterion itself extends without modification to the variable selection context, however the original support line procedure is in general inapplicable here because it assumes that the observed statistics are $p$-values, which are unavailable if the null distribution is not fully known. 
The challenging problem of controlling max-lfdr in finite sample under this model is studied in a companion paper (still in progress). 
In the current work, we concentrate on {\em estimating} the lfdr, defined under \eqref{eq:two-group} as
\begin{align}
\label{eq:lfdr}
    \lfdr(w) = \P(H_\idx=0 \mid W_\idx=w) = \pi_0 f_0(w) / f(w), 
\end{align}
where 
\begin{equation}
\label{eq:marginal}
f(w):= \pi_0f_0(w) + (1-\pi_0)f_1(w)
\end{equation}
is the marginal density. 
When $p$ is large and the observations are only weakly dependent, $\lfdr(w)$ roughly represents the fraction of observations near $w$ that arise from the null component $f_0$. 
Thus, the lfdr is not only a key ingredient for powerfully ordering hypotheses \citep[e.g.,][]{sun2007oracle, heller2021optimal}, but also directly relevant as a type I error criterion because it reveals the degree of heterogeneity in the frequency of errors throughout the rejection set \citep{efron2001empirical}. 

In general, further assumptions beyond null symmetry are needed to identify the $\lfdr$ function from the marginal distribution $f$ (this is discussed in the next section). The methodology we develop in this work is instead centered around estimating a conservative upper bound on the lfdr:
\begin{align}
\label{eq:clar}
\clar(w) := \frac{f(-w)}{f(w)}, \ \ \ \ \ w>0, 
\end{align}
which is clearly identifiable from the marginal distribution of $(W_\idx)$. Following a convention from the sparse signal detection literature \citep*{mccullagh2018statistical,xiang2024interpretation}, 
we call this density ratio the {\em complement of the local activity rate} (clar; see Remark \ref{remark:identifiability} on the terminology). 
To see why clar approximates lfdr, 
note that
\begin{align}
\label{eq:conservative-lfdr}
    \frac{\# \{j:W_j \approx -w\}}{\# \{j:W_j \approx w\}} &\geq \frac{\# \{j:W_j \approx -w, \beta_j=0\}}{\# \{j:W_j \approx w\}} \stackrel{(d)}{=} \frac{\# \{j:W_j \approx w, \beta_j=0\}}{\# \{j:W_j \approx w\}} ,
\end{align}
where the equality in distribution follows from the symmetric null assumption. 
Hence, we expect clar to be a conservative estimate of lfdr. If the non-null component is presumed not to contribute many large and negative contrast statistics, then the inequality above holds with little slack; in Appendix~\ref{sec:suppl_clar} we state a condition that implies $\clar(w)/\lfdr(w)$ is close to 1 for large values of $w>0$, and check it for several examples.

\begin{figure}[t]
    \centering
    \includegraphics[width=\linewidth]{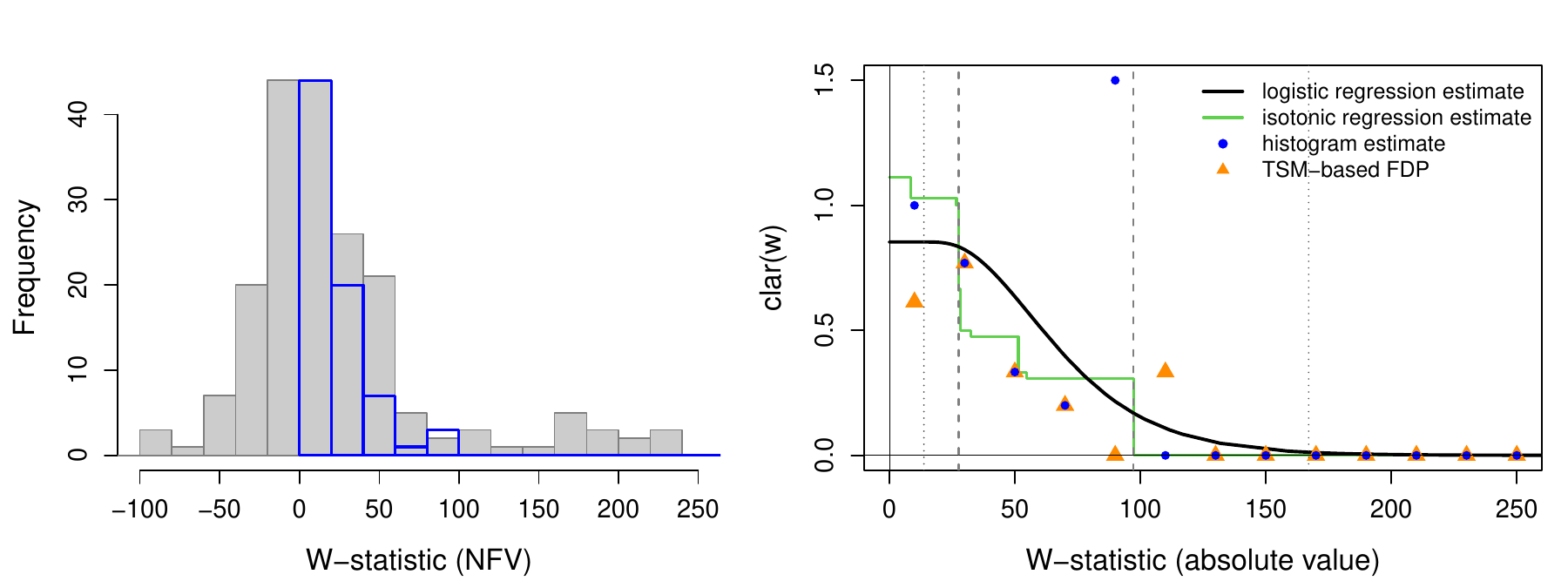}
    \caption{Left: Illustration of the histogram-based estimates on HIV data. The estimate of $\clar(w)=f(-w)/f(w)$ at $w>0$ is the height of the histogram at $-w$ divided by its height at $w$ (blue). Right: Logistic regression estimator with natural cubic spline basis (black). 
    Dashed vertical lines indicate internal knots, and dotted lines indicate the boundary knots; these are determined by the level sets of a preliminary isotonic regression estimator (green).}
    \label{fig:sign-logistic}
\end{figure}

A crude estimate of the clar is obtained simply by flipping the negative side of the histogram of the $W_j$ statistics about the $y$ axis, and dividing the height of this mirrored image by the height of the original histogram. The resulting estimates are illustrated in Figure \ref{fig:sign-logistic} for the HIV data example with $W$-statistics corresponding to the antiviral drug Nelfinavir (NFV).
The histogram estimate roughly tracks the true subset FDP in each bin, which was computed from treatment-selected mutation (TSM) panels \citep{rhee2005hiv}, shown in orange on the right panel. 
This crude estimate captures the essence of our approach; in the methodology part of the paper we mainly focus on developing refinements (smoother versions) of it. 
Specifically, we advocate a classification approach to estimating the ratio of densities in \eqref{eq:clar} (Section \ref{sec:clar-estimation}). 
Our main theoretical result, presented in Section \ref{sec:theory}, says that under suitable conditions,  the logistic regression method yields a consistent estimator of the corresponding population-level parameter, which can be thought of as the closest approximation to clar within the logistic model we consider. 
In addition, we establish general consistency results for other non-parametric or shape-constrained estimates of clar. 
The key condition needed for consistency of the estimator is 
convergence of the empirical cdf $F_p(w) \coloneqq \frac{1}{p}\sum_{\idx=1}^p 1\{W_\idx \leq w\}$ to the marginal cdf 
$F(w)$ in probability as $p \to \infty$ (Assumption \ref{assum:weak_dep}).

\section{The density ratio $f(-w)/f(w)$ as target of inference}
\label{sec:EB}

\subsection{A conservative surrogate for lfdr}
\label{subsec:surrogate}
In the standard two-groups model, where the null density $f_0$ is assumed known, $f_0(w)/f(w)$ is often used to upper bound the lfdr \eqref{eq:lfdr}, in which case it suffices to use an estimate of the marginal density $f(w)$ in order to estimate the lfdr (this can be combined with an estimate for $\pi_0$). 
The more adaptive approach that advocates using an {\em empirical null} \citep{efron2004large} uses the data to also model $f_0$,
usually as a Gaussian whose mean and variance are obtained through an estimate of the (log of) the marginal density $f(w)$ \citep{efron2008microarrays}. 
These considerations do not apply in the proposed two-groups model \eqref{eq:two-group}, where $f_0$ is unknown and typically non-Gaussian, calling for a different approach to estimating lfdr. 
While the case of an unspecified $f_0$ is indeed 
quite different, the clar quantity \eqref{eq:clar} provides a surrogate to lfdr that exploits the symmetry of $f_0$. 
The basic logic is simple: 
if we think of $W_j$ as generated from knockoffs, for example, and $\pi_0 \approx 1$ as in examples where sparse subset selection methods are typically used, then we expect that for sufficiently large $w>0$, most $W$-statistics near $-w$ correspond to true nulls. 
As a consequence $f(-w) \approx \pi_0 f_0(-w) = \pi_0 f_0(w)$, using the symmetry of $f_0$, which implies that $\lfdr(w) \approx f(-w)/f(w)$. 
In fact, this relation is always conservative, 
\begin{align}
\label{eq:clar-lfdr}
    \lfdr(w) = \frac{\pi_0 f_0(w)}{f(w)} = \frac{\pi_0 f_0(-w)}{f(w)} \leq \frac{f(-w)}{f(w)} = \clar(w), 
\end{align}
which is the population-level analogue to the informal calculation in \eqref{eq:conservative-lfdr}. 
The density ratio on the right-hand side of \eqref{eq:clar-lfdr} is a more convenient target than lfdr already because it involves only $f(w)$, the marginal density from which the observations directly arrive. 

\begin{figure}[t]
        \centering
        \includegraphics[width=0.6\linewidth]{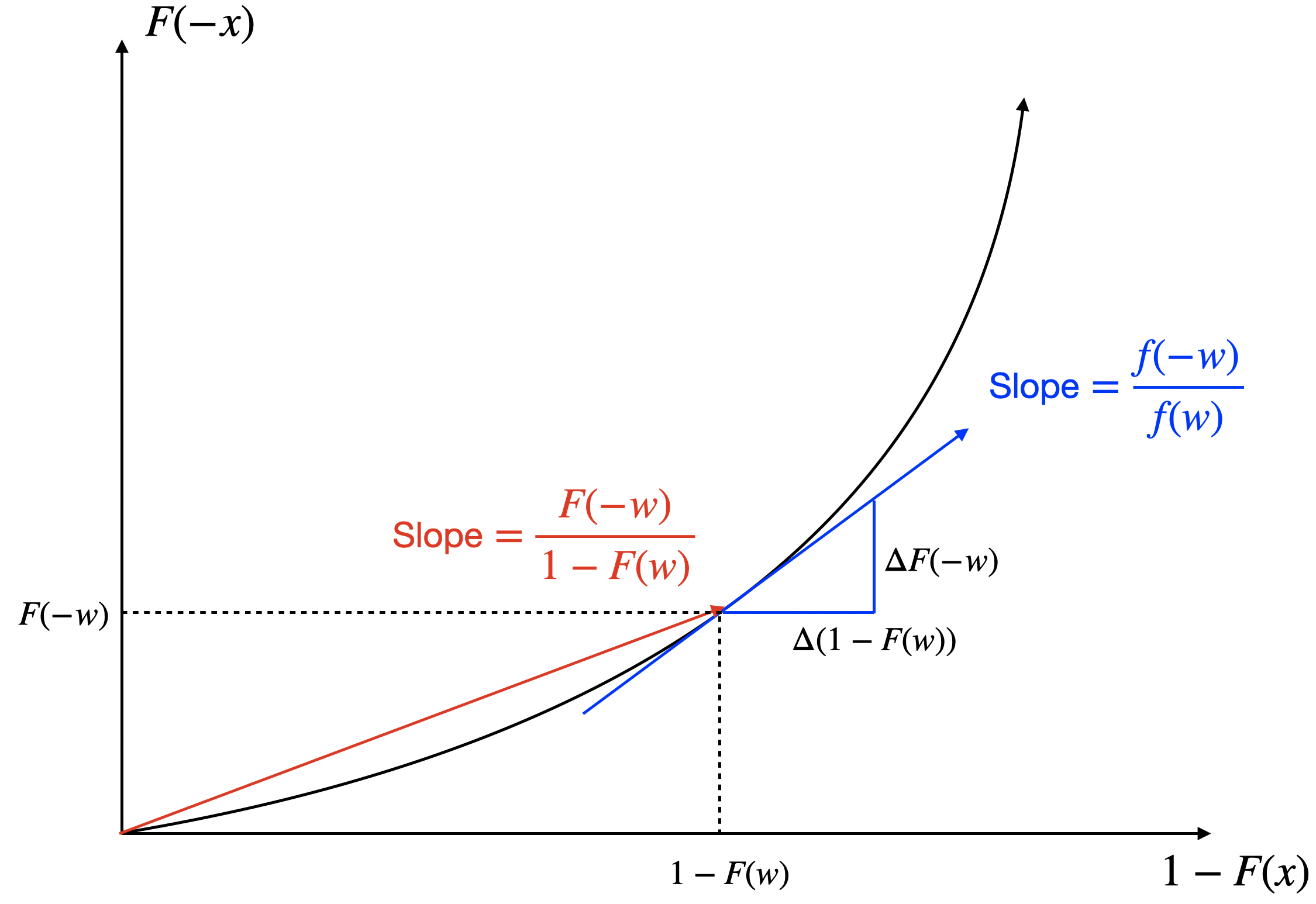}
        \caption{The slopes of secant and tangent lines correspond to the right-hand sides of \eqref{eq:clar-lfdr} and \eqref{eq:ctar-fdr}, respectively.}
        \label{fig:secant-tangent}
    \end{figure}
The display in \eqref{eq:clar-lfdr} can be viewed as a `local' analogue of the perhaps more familiar relation that uses the symmetry of $F_0$ to upper bound the {\em tail} false discovery rate,
\begin{align}
\label{eq:ctar-fdr}
    \Fdr(w) \coloneqq \P(H_\idx=0 \mid W_\idx \geq w) = \frac{\pi_0 (1-F_0(w))}{1-F(w)} = \frac{\pi_0 F_0(-w)}{1-F(w)} \leq \frac{F(-w)}{1-F(w)}, 
\end{align}
where $F(w)=\int_{-\infty}^w f(x) dx$ is the marginal cdf and $F_0$ is the cdf of the symmetric null distribution. 
Indeed, the relation \eqref{eq:ctar-fdr} is exploited, for example, by the knockoff filter, which substitutes the empirical cdf in place of $F$ on the right-hand side above to conservatively estimate the FDP among statistics $\geq w$.
The clar and the upper bound in \eqref{eq:ctar-fdr} have the same conditional averaging relationship as the local and tail FDR \citep[][Chapter 5]{efron2012large},
\begin{align}
\label{eq:sec-tan-fdr}
    \Fdr(w) &= \E\big[\lfdr(W_\idx) \mid W_\idx\geq w\big], 
\end{align}
which is illustrated geometrically in Figure~\ref{fig:secant-tangent}. 
\begin{figure}
\begin{tabular}{lll}
\includegraphics[width=0.31\linewidth]{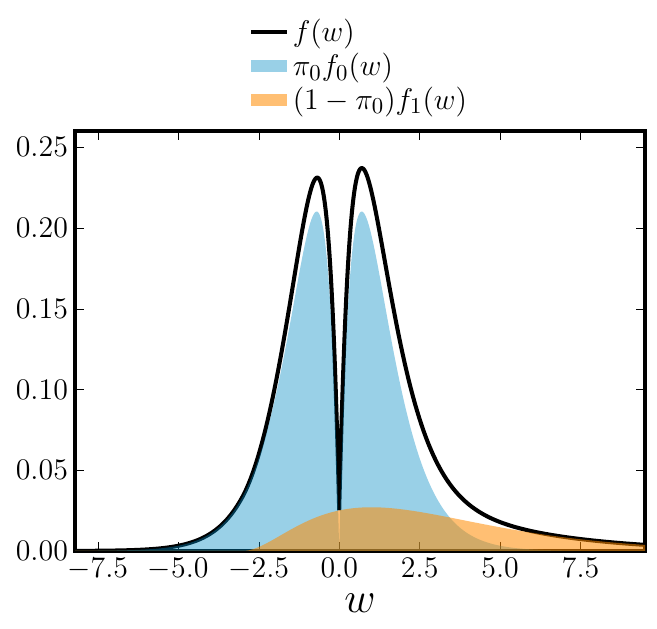} & \includegraphics[width=0.31\linewidth]{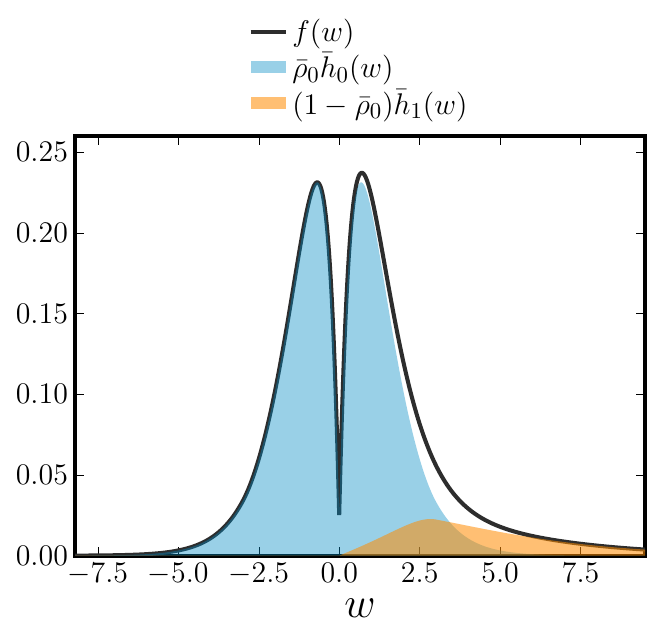} & \includegraphics[width=0.31\linewidth]{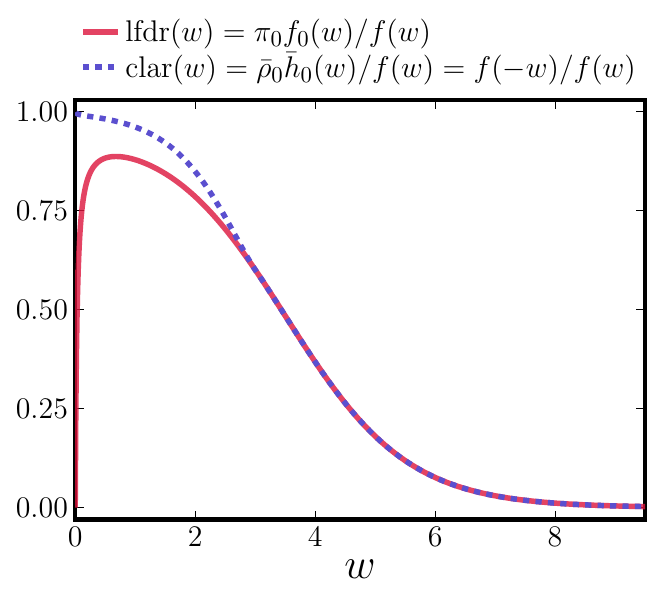}
\end{tabular}
\caption{Illustration of two decompositions of the marginal density $f$ with $\pi_0=0.8$, $f_0$ a symmetrized Gamma, and $f_1$ a shifted Gamma (more precisely, $f_0(w) = \{\Gamma(w; 2, 0.7)+\Gamma(-w; 2, 0.7)\}/2$ and $f_1(w) = \Gamma(w+3; 3,2)$, where $\Gamma(w;\alpha,\beta)$ denotes the density at $w$ of the Gamma distribution with shape $\alpha$ and scale $\beta$).
\textit{Left:} standard two-groups decomposition into null component $\pi_0 f_0(w)$ and non-null component $(1-\pi_0)f_1(w)$ as in~\eqref{eq:marginal}. \textit{Middle:} decomposition into maximally symmetric component 
$\bar{\rho}_0 \bar{h}_0(w) = f(w)\wedge f(-w)$ 
and its complement $(1-\bar{\rho}_0)\bar{h}_1(w) = f(w) - \bar{\rho}_0 \bar{h}_0(w)$
with $\bar{\rho}_0 
\approx 0.88$; $\bar{h}_1(w)=0$ for $w \in (-\infty, 0]$. \textit{Right:} the $\lfdr$ under the leftmost two-groups decomposition and the $\clar$ (which coincides with the $\lfdr$ for the middle two-groups decomposition). Note that $\lfdr(w) \leq \clar(w)$ for all $w \geq 0$ and that the inequality becomes tight for large $w$ (in this specific example). 
}
\label{fig:twogroups}
\end{figure}

Beyond being merely an upper bound to $\lfdr(w)$, $\clar(w)$ has appealing properties in terms of identifiability, sharpness, and probabilistic interpretation. The two-groups model in~\eqref{eq:marginal} is not identifiable: there are multiple choices of $\rho_0 \in [0,1]$ and densities $h_0,h_1$, with $h_0$ symmetric, that combine to the marginal density of the observations: $f = \rho_0 h_0 + (1-\rho_0)h_1$. Consequently, $\lfdr(w)$---which depends on the non-identifiable $f_0$ and $\pi_0$---is also not identifiable.\footnote{Concerns of non-identifiability apply more broadly to two-groups models even in the simpler case wherein the null density is known, see e.g.,~\citet{patra2016estimation, xiang2024interpretation} and references therein.
} 
Still, \citet[Theorem 1]{arias2021extending} prove that the particular choice 
$\bar{h}_0(w) \propto f(w) \wedge f(-w)$, which is the maximally symmetric component of $f$, leads to the decomposition $f=\bar{\rho}_0 \bar{h}_0 + (1-\bar{\rho}_0)\bar{h}_1$ with the largest possible $\rho_0=\bar{\rho}_0$. Unlike the original $(f_0,\pi_0)$,
the  maximally symmetric component $\bar{h}_0$ and its corresponding mixing weight $\bar{\rho}_0$ {\em are} identifiable. 
Moreover, under the assumption $f(-w) \leq f(w)$ for all $w \geq 0$, it can be verified that $\clar(w)$  {\em coincides} with the $\lfdr$ for the 
two-groups model with maximally symmetric null. 
In other words, $\clar(w)$ is sharp: equality in~\eqref{eq:clar-lfdr} is attained by a member of the class of (partially) identified two-groups models, which is itself identifiable;
Figure~\ref{fig:twogroups} illustrates these ideas. Appendix~\ref{sec:suppl_clar} elaborates on the above and, following~\citet{mccullagh2018statistical, xiang2024interpretation}, also provides a probabilistic interpretation of $\clar(w)$ in terms of {\em activity} indicators (which explains its name).

\subsection{Estimating clar}
\label{sec:clar-estimation}

The relation \eqref{eq:clar-lfdr} suggests a direct approach for conservatively estimating lfdr, namely, by estimating the density ratio in \eqref{eq:clar}. 
This could be done by approximating $f$ with a kernel density estimator $\hat{f}$, and then taking the estimate to be $\hat{f}(-w)/\hat{f}(w)$. Consistency holds for this approach under a weak dependence condition as $p \to \infty$ (Section \ref{sec:theory}), but if $f(w)$ is small at a point $w$ of interest then division by a noisy estimate $\hat{f}(w)$ can introduce unnecessary instability.

We propose instead to estimate the density ratio directly, leveraging a connection to supervised classification \citep[e.g.,][]{sugiyama2012density,efron2001empirical}. 
Consider the binary classification problem where the labels are the signs of the observed statistics $\sgn(W_\idx)\in \{\pm 1\}$, and the features are the absolute values $|W_\idx|$. 
Treating the negative label as ``success", 
the conditional odds are  
\begin{align}
\label{eq:odds-relation}
\frac{\P(W_\idx < 0 \mid |W_\idx|=w)}{1-\P(W_\idx < 0 \mid |W_\idx|=w)}
=
\frac{f(-w)}{f(-w) + f(w)} \cdot \frac{f(-w)+f(w)}{f(w)}
=
\frac{f(-w)}{f(w)}, 
\end{align}
i.e., this is $\clar(w)$. 
The binary classification framing allows us to impose 
structural constraints on the density ratio that can be easier to justify than specific 
smoothness assumptions about $f$ itself. 

Since we want  flexibility in modeling the logit of $\P(W_\idx < 0\mid |W_\idx|=w)$ as a function of $w$, a convenient way to proceed is to translate each $|W_\idx|$ into a covariate vector $Z_\idx = h(|W_\idx|)$ through some prespecified feature map $h:\R_+ \to \R^J$, and fit a logistic regression of $\mathbf{1}\{W_\idx<0\}$ onto $Z_\idx$.  A polynomial basis of high enough degree can effectively approximate any positive continuous odds function (see Appendix \ref{sec:proofs} for a formal statement), but may extrapolate poorly beyond the range of the data. Instead, we use a natural cubic spline basis, which extrapolates linearly beyond the right boundary knot. 
Beyond the left boundary knot, we impose a zero derivative constraint. Knot placement and other implementation details can be found in Appendix \ref{sec:spline-implementation}.

The zero derivative constraint prevents an undesirable scenario that arises in the unconstrained fit, namely, 
a wiggle often appears in the fitted clar function near the origin due to the flexibility of the cubic spline there. 
The true clar typically has a small negative slope near $w=0$ (e.g.~Figures \ref{fig:indep-sim} and \ref{fig:count-ratio}), so the constraint only introduces a minor bias in exchange for smoothing out this behavior.
Overall, we find that the logistic method
produces well-behaved estimates across a range of settings, as illustrated through simulation experiments in Section \ref{sec:numerical} and on real data in Section~\ref{sec:applications}.

\section{Asymptotic analysis}
\label{sec:theory}

\subsection{Estimation consistency} 
\label{sec:supervised-classification}

We state asymptotic consistency results for three methods of estimating $\clar$: the logistic regression classifier, the KDE method, and an isotonic regression estimator. Our results do not assume independence between $W_1,\dots,W_p$, and instead require convergence of the empirical cdf  (ecdf) to the marginal cdf $F$ as $p \to \infty$. 
\begin{assumption}[Weak dependence]
Let $F$ denote the marginal cdf of each of $W_1,\dots,W_p$, let $F_p(w)=\frac{1}{p}\sum_{\idx=1}^p \mathbf{1}\{W_\idx\leq w\}$ denote their ecdf, and let $F_{0p}(w)=\frac{1}{p_0}\sum_{\idx=1}^p \mathbf{1}\{W_\idx \leq w,\; H_\idx=0\}$ denote the ecdf among nulls.  We assume that:  $\|F_p-F\|_\infty \to 0$ and $\|F_{0p}-F_0\|_\infty \to 0$ in probability as $p \to \infty$, where $F_0$ is the null cdf.
\label{assum:weak_dep}
\end{assumption}

Assumption \ref{assum:weak_dep} commonly appears in the multiple testing literature; see, for example, Condition~2 of Theorem 1 in \cite{fan2025asymptotic} and the `weak dependence' condition in \cite{storey2004strong}. A common form of this condition is uniform convergence of the empirical distribution of the nulls. 
We additionally require the empirical distribution of the full sample to converge uniformly to the marginal cdf. 
We illustrate the convergence of $F_p$ to $F$ (and of $F_{0p}$ to $F_0$) empirically for several types of knockoff statistics in Figure \ref{fig:ecdf-diagnostic1}.

Lemma \ref{lem:pointwise-uniform} in Appendix \ref{sec:proofs} states that the uniform convergence in Assumption \ref{assum:weak_dep} is implied by pointwise convergence to a continuous cdf. While the proofs for results stated in this section assume uniform convergence as a starting point, pointwise convergence is sufficient in the current setting, since the existence of a density implies $F$ is continuous.

\vspace{1em}

\noindent \textit{Consistency for logistic regression.} 
Let $X_\idx = h(|W_\idx|) \coloneqq (h_1(|W_\idx|),\dots,h_{J}(|W_\idx|))$ be a featurization of $|W_\idx|$, for example a natural cubic spline basis. 
Define the logistic objective and a corresponding maximizer:
\begin{align*}
    \widehat{M}_p(\beta) &\coloneqq \frac{1}{p}\sum_{\idx=1}^p \Big[ \beta^\top X_\idx\; \mathbf{1}\{W_\idx<0\}   -  \log (1+e^{\beta^\top X_\idx}) \Big], \;\; \;
    \hat{\beta} \in \argmax_{\beta \in \R^{J}} \widehat{M}_p(\beta),
\end{align*}
and define the population-level analogs, 
$
M^*(\beta):=  \E[\widehat{M}_p(\beta)] 
$ and $\beta^* \in \argmax_{\beta \in \mathbb R^J} \; M^*(\beta)$. The latter is unique if we assume the following two conditions.
\begin{assumption}
\label{assumption:1}
$0 \prec \E[X_\idx X_\idx^\top] \prec \infty$, i.e.~$\E[X_\idx X_\idx^\top] \succ 0$ and every entry of $\E[X_\idx X_\idx^\top]$ is finite.
\end{assumption}
\begin{assumption}
    \label{assumption:2}
    $\P[ (2\; \mathbf{1}\{W_\idx < 0\}-1)\beta^\top X_\idx  \geq 0] <1$ for all $\beta \in \mathbb R^J \setminus \{0\}$.
\end{assumption}
Assumption \ref{assumption:1} implies the population-level objective is concave in $\beta$, and Assumption \ref{assumption:2} requires that no linear predictor perfectly distinguishes the sign of $W_\idx$ almost surely, which is implied for instance by the condition $\pi_0 f_0(w)>0$ for all $w$ (see Lemma \ref{lem:basic-non-separability} in Section \ref{sec:proofs}).
Our next result establishes consistency of $\hat{\beta}$ for $\beta^*$ under regularity conditions. The proof can be found in Section~\ref{sec:proofs}. 

\begin{theorem}
\label{thm:glm-consistency}
    Suppose Assumptions~\ref{assum:weak_dep}, \ref{assumption:1}, \ref{assumption:2} hold, and $h_1,\dots,h_J$ are continuous and bounded functions. Then $\hat{\beta}$ exists with probability tending to $1$ and $\hat{\beta} \to \beta^*$ in probability as $p\to\infty$. 
\end{theorem} 

\vspace{1em}

\noindent \textit{Consistency for isotonic regression.} The Grenander estimator of clar is defined as follows. Let $y_\idx = \mathbf{1}\{W_\idx<0\}$ for $i=1,\dots,p$, and sort them in decreasing order of magnitude $|W|_{(1)}\geq \dots \geq |W|_{(p)}$. Let $\hat{\mu}_{(1)},\dots,\hat{\mu}_{(p)}$ denote the isotonic regression estimator computed from the sign indicator values $y_{(1)},\dots,y_{(p)}$, as obtained  from the left-hand slopes of the greatest convex minorant of the partial sum process $k \mapsto \sum_{i=1}^k y_{(i)}$. In R, one can run \texttt{pava()} on $y_{(1)},\dots,y_{(p)}$ to obtain $\hat{\mu}_{(1)}\leq \dots \leq \hat{\mu}_{(p)}$. The estimate of clar at $|W|_{(k)}$ is the odds $\hat{\mu}_{(k)}/(1-\hat{\mu}_{(k)})$. We extend this estimate to arbitrary $w>0$ as follows:
\begin{align}
\label{eq:isotonic-clar}
    \widehat{\clar}_{\text{iso}}(w) \coloneqq \frac{\hat{\mu}(w)}{1-\hat{\mu}(w)},
\end{align}
where $\hat{\mu}(w) \coloneqq \hat{\mu}_{(k)}$ when $w \in (|W|_{(k+1)},|W|_{(k)}]$, where $|W|_{(0)}:=\infty$ and $|W|_{(p+1)}:=0$, taking the convention that $\widehat{\clar}_{\text{iso}}$ is a left-continuous piece-wise constant function. Under strict monotonicity and weak dependence, the isotonic regression estimator is consistent for the true clar. This result is summarized below and proved in Section \ref{sec:proofs}. 
\begin{theorem}
\label{thm:consistency-iso}
Suppose $f$ is a positive and continuous probability density function, the true regression function $\mu(w) = \P(W_\idx < 0 \mid |W_\idx|=w)$ is non-increasing in $w$, and Assumption~\ref{assum:weak_dep} holds. Then for each $w>0$, $\widehat{\clar}_{\textnormal{iso}}(w) \to \clar(w)$ in probability as $p \to \infty$.

\end{theorem} 

\vspace{1em}

\noindent \textit{Consistency for kernel density estimate (KDE).}
Herein, we describe a simpler strategy for estimating $\clar(w)$: we separately estimate the numerator and denominator in~\eqref{eq:clar}. Fix $w>0$ and let $\hat{f}(w)$, resp.\ $\hat{f}(-w)$ be estimates of $f(w)$, resp.\ $f(-w)$. Then we can estimate $\clar(w)$ as $\hat{f}(-w)/\hat{f}(w)$. To instantiate this idea more concretely, suppose $\hat{f}(w)\equiv \hat{f}_h$ is a kernel density estimate with kernel $K$ and bandwidth $h$, that is, 
$$
\hat{f}_h(w) := \frac{1}{ph} \sum_{j=1}^p K\left(\frac{W_j-w}{h}\right),\;\;\; \widehat{\clar}_h(w) := \frac{\hat{f}_h(-w)}{\hat{f}_h(w)}.
$$
Under the weak dependence condition (Assumption \ref{assum:weak_dep}), the KDE method remains a consistent estimator for the density ratio $f(-w)/f(w)$.  
This result is stated below and proved in Section 
\ref{sec:proofs}.

\begin{theorem}
\label{thm:KDE-consistency}
Fix $w>0$. Suppose that the density $f$ is continuous and positive in a neighborhood of $w$ and $-w$, and Assumption \ref{assum:weak_dep} holds. Also, assume that the kernel $K$ is supported on $[-1,1]$, is Lipschitz continuous on $\mathbb R$ and satisfies $\int_{[-1,1]}K(u) du = 1$. Then there exists a sequence $h_p \to 0$ such that \smash{$\widehat{\clar}_{h_p}(w) \to \clar(w)$} in probability as $p \to \infty$.
\end{theorem}

\subsection{Asymptotic boundary FDR control}
In practice, it is common to shortlist promising hypotheses by selecting those with $W_\idx$ greater than some critical cutoff $\hat{\tau}$, where larger positive values of $W_j$ are generally associated with stronger evidence against the null. 
Consider setting this cutoff to ensure that rejections near $\hat{\tau}$
are estimated to have lfdr just below~$\alpha$. \cite{soloff2024edge} showed this approach controls the expected (true) lfdr $\leq \alpha$ at the rejection threshold when using a particular shape-constrained estimator of lfdr \citep{strimmer2008unified,grenander1956theory}, assuming $f_0$ is known and the test statistics are independent.
When $f_0$ is unknown but symmetric (and the non-nulls are positively enriched), we set the threshold according to an estimate of clar: 
\begin{align}
\label{eq:clar-hat-cutoff}
    \hat{\tau} \coloneqq \inf\{w \geq 0 : \widehat{\clar}(w) \leq \q\}.
\end{align} 
Near the cutoff, the false discovery proportion among boundary rejections is controlled below $\alpha$ asymptotically under regularity conditions (Theorem \ref{thm:asymptotic-bFDR}). 
The intuition is as follows: when the estimator $\widehat{\clar}$ is uniformly consistent in a neighborhood of the population-level target, $\clar(\hat{\tau})$ converges to $\alpha$ as $p\to\infty$ under monotonicity and continuity. Since clar is a conservative surrogate for lfdr, the latter is asymptotically $\leq\alpha$ among rejections near $\hat{\tau}$. 
A formal proof is provided in Appendix~\ref{sec:proofs}.

\begin{theorem}
\label{thm:asymptotic-bFDR}
Fix $\alpha \in (0,1)$ for which there is a unique $w^*$ satisfying $\clar(w^*)=\alpha$, and suppose that
\begin{enumerate}[(i)]
    \item $\clar$ is continuous at $w^*$ and non-increasing on $[0,w^*+\eta)$ for some $\eta>0$, 
    \item $f,f_0$ are positive and continuous in a neighborhood of $w^*$, and $f_0$ is symmetric about zero, 
    \item $p_0/p\to \pi_0\in(0,1)$ as $p\to\infty$, where $p_0$ is the number of true nulls, 
    \item Assumption \ref{assum:weak_dep} holds, i.e.~$\|F_p-F\|_\infty$ and $\|F_{0p}-F_0\|_\infty$ both tend to zero in probability as $p\to \infty$,
    \item $\widehat{\clar}$ is uniformly consistent for $\clar$ on an open interval containing $w^*$, and one of the following conditions holds: $\widehat{\clar}$ is non-increasing on $[0,w^*)$ with probability tending to $1$, or $\widehat{\clar}$ is uniformly consistent for $\clar$ over $[0,w^*)$.
\end{enumerate}
Then, for the threshold $\hat{\tau}$ defined in \eqref{eq:clar-hat-cutoff}, 
    there exists a sequence of positive numbers $\varepsilon_p \to 0$ such that for any $\delta>0$,
\begin{align*}
\lim_{p\to\infty}\P\big(\FDP([\hat{\tau},\hat{\tau}+\varepsilon_p]) > \alpha + \delta \big) = 0.
\end{align*}
\end{theorem}

\section{Numerical experiments}
\label{sec:numerical}

We turn to some simulation experiments for assessing the accuracy of the proposed methods. 
In the first example we consider an independent two-groups model as a simple testing ground where the behavior of the methods, in particular the proposed logistic regression estimator, is more predictable. 
In the second, more realistic example the observations are $W$-statistics from a fixed-$X$ knockoffs implementation. 

\medskip
\noindent {\bf Independent two-groups model.} 
In this experiment the statistics are drawn  independently from a simple two-point Gaussian mixture,  
\begin{align}
\label{eq:numerical-two-group}
    W_\idx \; \stackrel{\iid}{\sim} \; 0.9 \; \calN(0,1) + 0.1\; \calN(\mu=2.5,1), 
\end{align}
for $\idx=1,\dots,p=10^4$. 
We compare the following methods for estimating clar:
\begin{itemize}[leftmargin=*]
    \item {\em Isotonic method (Grenander)}. This method was described in Section \ref{sec:supervised-classification}. Briefly, we order the indicator variables $y_\idx \coloneqq \mathbf{1}\{W_\idx<0\}$ by $|W|_{(1)}\geq \dots\geq |W|_{(p)}$, and run \texttt{pava()} on $y_{(1)},\dots,y_{(p)}$. The Grenander estimate of  clar($|W|_{(k)})$ is the fitted conditional odds $\hat{\mu}_{(k)}/(1-\hat{\mu}_{(k)})$, where $\hat{\mu}_{(k)}$ is the $k^{\text{th}}$ fitted value from isotonizing $y_{(1)},\dots,y_{(p)}$.
    \item {\em Kernel density estimation (KDE)}. This approach was described at the start of Section \ref{sec:clar-estimation} (and implicitly in Section 2 of \cite{arias2021extending}). 
    We fit the marginal density of the $W$-statistics using the \texttt{density()} function with default settings in R, and take $\widehat{\clar}(w) = \hat{f}(-w)/\hat{f}(w)$. 
    \item {\em Natural splines logistic classifier}. This is the method proposed toward the end of Section~\ref{sec:clar-estimation},  employing a natural cubic spline basis that extrapolates linearly beyond the range of the $|W_\idx|$'s. 
\end{itemize}
\begin{figure}[h]
\centering
\includegraphics[width=\linewidth]{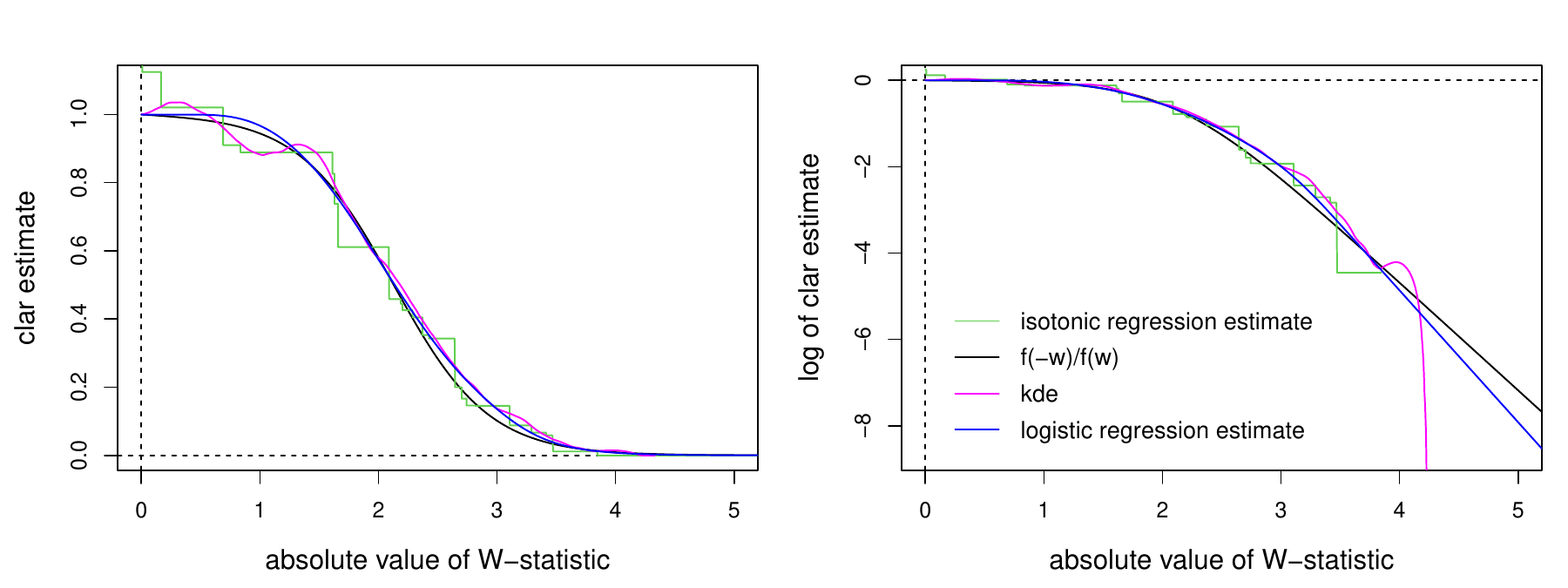}
\caption{Comparison of clar estimates for several methods in the Gaussian two-groups model with location shift. }
\label{fig:indep-sim}
\end{figure}
The left panel of Figure \ref{fig:indep-sim} shows the estimate of $\clar(|W_\idx|)$ against $|W_\idx|$ for each of the estimators mentioned above, along with the true ratio, $\clar(w) = f(-w)/f(w)$ for $w>0$. 
The right panel displays the same plot, except with the $y$-axis shown on a logarithmic scale to allow examining performance of the different  methods at the tail (large $|W_\idx|$). 
The plots demonstrate that the logistic method performs  well in estimating the true density ratio; 
in particular, the choice of the natural spline basis seems to provide appropriate regularization, which helps avoid the wiggle apparent in the other methods. 
To further support the appropriateness of the linear extrapolation (of the natural spline basis) beyond the boundary knots, which is more evident in the right panel of the figure, consider the true log odds function in this setting, 
\begin{align*}
    \log \frac{f(-w)}{f(w)} = \log \Big(\frac{1+\frac{1-\pi_0}{\pi_0} \phi(w+\mu)/\phi(w)}{1+\frac{1-\pi_0}{\pi_0} \phi(w-\mu)/\phi(w)}\Big). 
\end{align*}
As $w\to \infty$, the log odds at $w$ satisfies 
\begin{align}
\label{eq:limit-log-odds}
    \lim_{w \to\infty}\frac{\log f(-w)/f(w)}{-\mu w}=1, 
\end{align}
so 
$\text{logit}\; \P(W_\idx < 0 \mid |W_\idx|=w)$ is indeed expected to be approximately a linear function of $w=|W_\idx|$ when $w\gg \mu$. 
For positive $w$ near zero, the odds ratio behaves as
\begin{align*}
    \frac{f(-w)}{f(w)} \approx 1 + \frac{1-\pi_0}{\pi_0} \big( e^{-w\mu-\mu^2/2} - e^{w\mu - \mu^2/2} \big) \sim 1-\frac{1-\pi_0}{\pi_0} \; 2w\mu,
\end{align*}
but can deviate from this linear approximation as $w=|W_\idx|$ moves further away from zero. 
At the same time, the logistic method picks up the correct overall trend of the true curve. 
Importantly, it is able to roughly locate the steep decline around $w\approx 1.5$ which is indicated also in  the isotonic estimate; this supports the heuristic of building on the isotonic fit to set the knots for the logistic method in the implementation, as detailed in Appendix \ref{sec:spline-implementation}.

\medskip
\noindent {\bf Knockoff $W$-statistics.} 
Our second example simulates a variable selection problem where the $W_\idx$ are knockoff $W$-statistics. 
The setting is adapted from the \href{https://web.stanford.edu/group/candes/knockoffs/software/knockoffs/tutorial-3-r.html}{knockoffs web tutorial}. 
We simulate the raw data $(\bX, Y)$ from a Gaussian linear model with $n=3000$ cases and $p=1000$ predictors, among which $s=100$ have coefficient $\beta_j=4.5/\sqrt{n}$ and the rest of the coefficients are zero. 
The rows of the matrix $\bX$ are generated from a multivariate normal distribution with mean zero and covariance $\Sigma_{ij} = \rho^{|i-j|}$ for $\rho=0.25$. 
The $W_\idx,\; \idx=1,...,1000$, are {\em Lasso signed-max} statistics \citep[LSM,][]{barber2015controlling} calculated on the knockoffs-augmented matrix where $\bX$ is regarded as fixed. 
\begin{figure}[t]
    \centering
    \includegraphics[width=\linewidth]{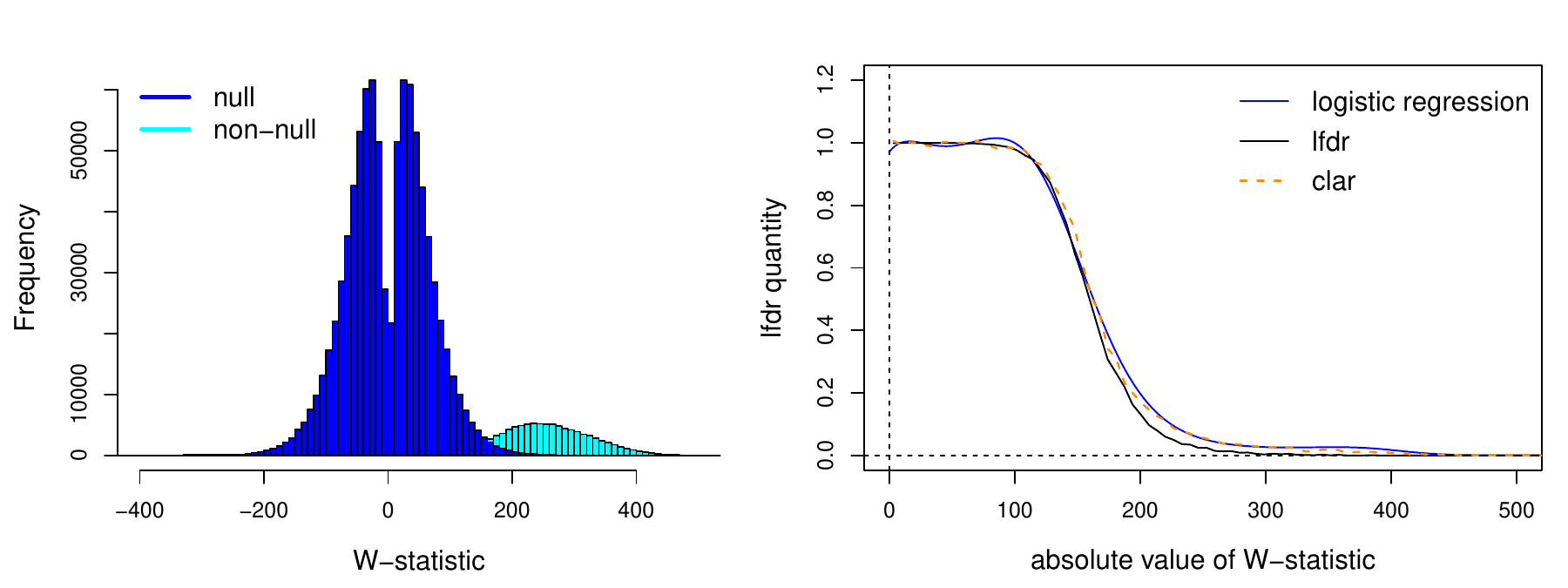}
    \caption{On the left, a histogram of Lasso-sign-max (LSM) importance statistics is shown, obtained from $N=1000$  repeated runs of a variable selection experiment with $p=1000$ variables.
    On the right, the clar is estimated 
    using logistic regression, the lfdr is estimated using a fixed binwidth and 100 equally spaced bins, and the true clar is estimated by the odds ratio for the proportion of absolute value $W$-statistics in each bin that are negative.}
    \label{fig:count-ratio}
\end{figure}
The left panel of Figure \ref{fig:count-ratio} shows the marginal distribution of the $W$-statistics, gathered from $N=1000$ repetitions of this experiment. A logistic regression estimate of clar is shown on the right panel. For reference, the plots also show approximations of the true clar and the true lfdr, obtained by averaging the suitable local count ratios over repeated runs of the experiment. Figure \ref{fig:sim-knockoffs-clar} shows the clar estimate produced by the (natural splines) logistic method from a single repetition. The logistic estimate is a bit conservative in this example, but still able to capture the overall shape in the graph of the true clar. 
\begin{figure}[h]
    \centering
    \includegraphics[width=.5\linewidth]{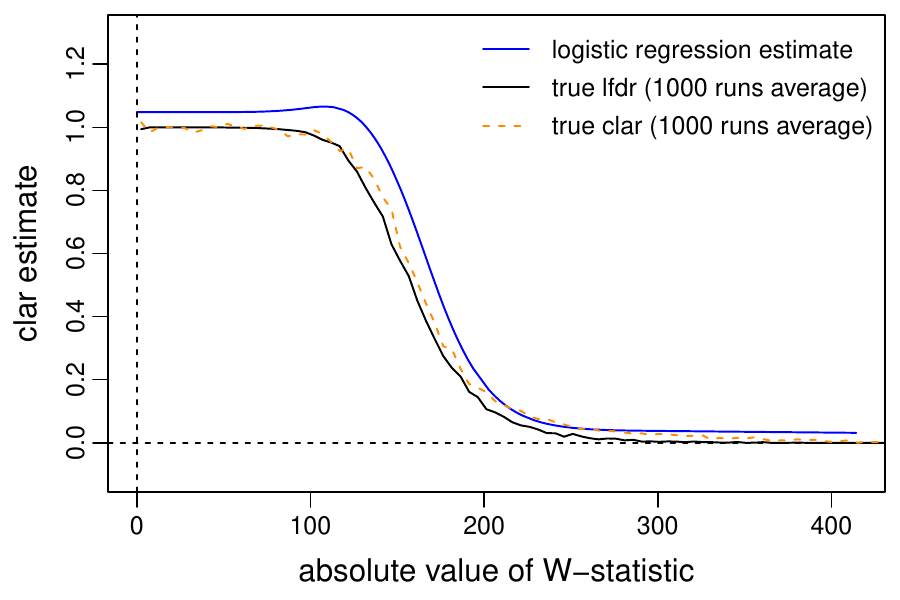}
    \caption{Illustration of the cubic spline logistic method for clar estimation in a knockoffs setting. 
To approximate the true clar, in each bin we calculated the count of negative $W_\idx$ in the mirror image region, divided by the count of positive $W_\idx$ in the original bin. Similarly, for each of the 1000 independent runs we divided the absolute values of the realized $W_\idx$ into small bins, calculated the (true) FDP in every bin as an estimate of the local FDP, and reported the average FDP in each bin. We then took the average in each bin over the 1000 runs and used the relationship \eqref{eq:logistic} to transform this into an estimate of clar and lfdr, respectively.} 
    \label{fig:sim-knockoffs-clar}
\end{figure}

\subsection{Standard error estimates}
\label{sec:standard-error}
We now describe two methods for obtaining standard errors for the logistic regression estimator of $\clar$. The first method is based on robust sandwich estimators of the variance, which are converted to standard errors for $\widehat{\clar}$ through the delta method. 
The second method is based on a parametric bootstrap-after-selection method developed by \cite{luo2024estimating} for assessing variability in FDR estimates of variable selection procedures.

\paragraph{Sandwich standard errors.} Parameter estimates for the logistic model coefficients $\gamma_k$ in \eqref{eq:logistic}, or the more practical cubic spline version, are computed by the \texttt{glm()} function in R, where we define the response as $1\{W_\idx < 0\}$, and the predictors as polynomial basis functions evaluated at $|W_\idx|$. 
Approximate standard errors can be calculated using the delta method: given the fitted coefficients $\hat{\gamma}_0,\dots,\hat{\gamma}_J$ and an estimate of their covariance matrix $\widehat{\Sigma}$, the variance of the fitted value $g(\hat{\gamma}) \coloneqq \exp\Big\{\sum_{k=0}^J \hat{\gamma}_k h_k(w) \Big\}$ is approximated as
\begin{align*}
    \widehat{\operatorname{Var}}[g(\hat{\gamma})] \approx \nabla g(\hat{\gamma})^\top \widehat{\Sigma}\; \nabla g(\hat{\gamma}), \hspace{2em} [\nabla g(\gamma)]_k \coloneqq h_k(w) \exp\Big\{ \sum_{j=0}^J \gamma_j h_j(w) \Big\}.
\end{align*}
Since the logistic model for estimating $\clar$ is misspecified even in ideal settings such as the 
i.i.d.~Gaussian two-groups model, this standard error procedure can benefit from using a robust (sandwich) estimator $\widehat{\Sigma}$ for the covariance matrix; such can be computed readily in R using the \texttt{vcovHC} function in the \texttt{sandwich} package. 
For independent observations $W_1,\dots,W_p$, this  approach ensures asymptotic coverage  of the closest representative in the logistic model to the truth (measured in KL distance). 
However, delta method standard errors can be grossly invalid when the $W_\idx$ are not independent, which is the case when $W_\idx$ are obtained with knockoffs, for example. 
In Figure \ref{fig:coverage-naive}, the right panel plots the delta method intervals 
for the numerical experiment from Section \ref{sec:intro}, where the $W_\idx$ are LSM statistics calculated using knockoffs. 
In the left panel of Figure \ref{fig:coverage-naive} we assess the coverage rate of the delta method intervals in 1000 runs of the entire  experiment; the intervals target 95\% coverage but achieve much lower coverage for values between $200 \leq w \leq 300$.

\paragraph{Bootstrap standard errors.}

To mitigate the limitations of the delta method, we offer an alternative, bootstrap approach. 
For the independent two-groups example \eqref{eq:numerical-two-group}, which involves no covariates, using the usual bootstrap scheme  produces the standard error estimates reported in the left panel of Figure \ref{fig:sim-indep-se}  (the plot displays $\pm 2$ SE around the same clar estimate from the original  realization in Figure \ref{fig:indep-sim}). 
For verification, the right panel of the figure plots true standard errors, as approximated from $N=200$ fresh Monte Carlo repetitions of the experiment, which show good agreement with the bootstrap estimates in the left panel. 

\begin{figure}[h]
\centering
\includegraphics[width=\textwidth]{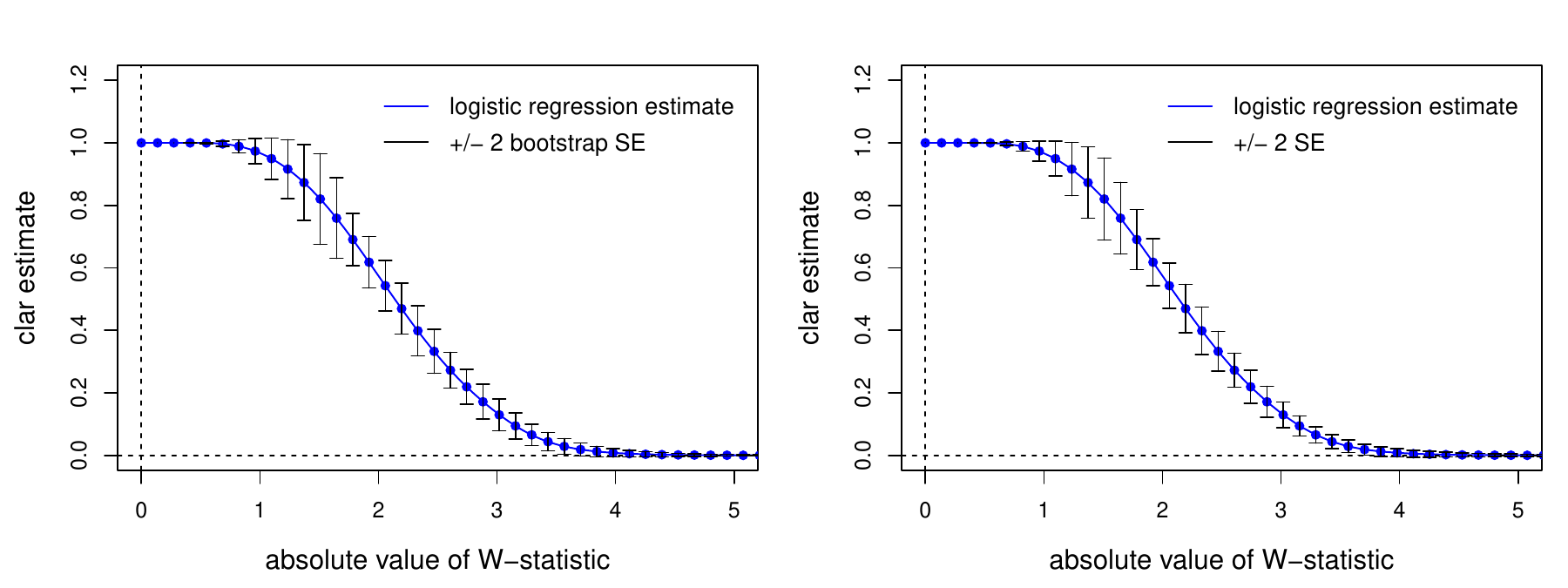}
\caption{Standard errors for the independent two-groups example. 
{\em Left}: bootstrap estimates of the SE. {\em Right}: true standard errors as estimated from 200 fresh draws.}
\label{fig:sim-indep-se}
\end{figure}

\begin{figure}[h]
\centering
\includegraphics[width=\linewidth]{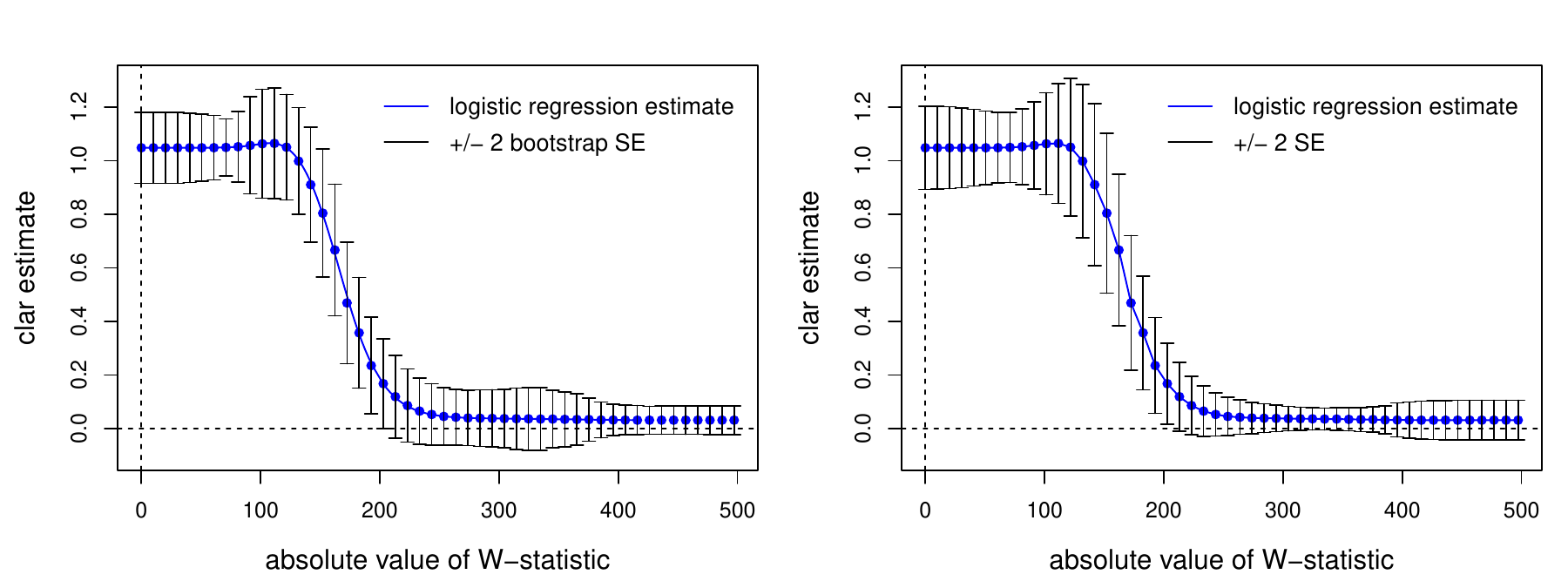}
\caption{Standard errors for the knockoffs example. 
Left panel shows parametric bootstrap estimates of the SE, right panel shows `true' standard errors as estimated from  1000 fresh draws. 
}
\label{fig:sim-kf-se}
\end{figure}

However, when $p$ is large relative to $n$, this approach suffers from the curse of dimensionality---the bootstrap distribution fails to represent the true distribution of the point estimate. 
To demonstrate this, Figure \ref{fig:OLS-statistics} in the Appendix shows histograms of the true $W_\idx$ statistics versus the histogram produced by pairs bootstrap, for simulation examples where $n=10^4, p=400$ and the sparsity is 0.1, i.e. $\beta$ has $s=40$ nonzero entries, each of which is equal to $4.5/\sqrt{n} = 0.045$. 
Here we used OLS instead of LSM for computing the $W$-statistics, and the  rows in the figure correspond to different settings for the $X$ matrix. 
The pairs bootstrap histogram clearly looks different from the true histogram already when $X$ has independent columns. 
The situation is much worse when the columns of $X$ are correlated: in that case the histogram from the pairs bootstrap shows a heavy tail on the left, which (as depicted in the figure) contained mostly {\em non}-nulls. 
In Figure \ref{fig:pairs-bootstrap-low-dim} we repeated the experiment in the middle row of Figure \ref{fig:OLS-statistics}, but for $p=40$ and $s=10$ instead of $p=400,s=40$ ($n=10^4$ as before). 
In that case, where $p$ is very small compared to $n$, pairs bootstrap performed very well in estimating the true histogram. 
This suggests that the problem is indeed related to the size of $p$ relative to $n$. 

To overcome this issue, we propose to use a {\em parametric} bootstrap approach instead: first we fit Lasso regression to select a subset of predictor variables. We then fit an ordinary least squares (OLS) regression using only the selected variables, obtaining coefficients $\hat{\beta}^{\text{OLS}}$ and residual standard error $\hat{\sigma}^{\text{OLS}}$. 
Finally, bootstrap samples of the response vector are generated as
\begin{align*}
    y^{(b)} \sim \mathcal{N}_n(X \hat{\beta}^{\text{OLS}}, (\hat{\sigma}^{\text{OLS}})^2 I), \quad  b=1,\dots,B. 
\end{align*}
For each bootstrap dataset we recompute the statistics $W_1,\dots,W_p$ and estimate $\clar$; 
the standard deviation of $\widehat{\clar}_b(w)$ across bootstrap samples $b=1,\dots,B$ gives the standard error estimate, shown in the left panel of Figure \ref{fig:sim-kf-se} around the same clar estimate from the single realization in Figure \ref{fig:sim-knockoffs-clar}). 
These compare well with the true standard deviations, which we estimated based on $N=200$ fresh Monte Carlo replications and displayed in the right panel of the figure, 
except for the tail (large $|W_\idx|$) where they are not reliable. 
This is further confirmed by the right column in Figure \ref{fig:OLS-statistics}, which demonstrates that the parametric bootstrap approach produces much better estimates for the true distribution of the $W$-statistics compared to pairs bootstrap.

\section{Applications}
\label{sec:applications}

In this section, we apply our methodology for clar estimation on two real datasets from scientific contexts where the symmetric null assumption can be reasonably justified. The first dataset comes from a proteomics experiment where the goal is to identify amino-acid chains useful for inhibiting protein binding, and the second one is the HIV drug-resistance dataset previously analyzed by \cite{barber2015controlling}, where the knockoffs construction guarantees symmetry under the null. 

\subsection{High-throughput profiling of reactive cysteines}

\begin{figure}[h]
    \centering    \hspace{-1.5em}\includegraphics[width=1\linewidth]{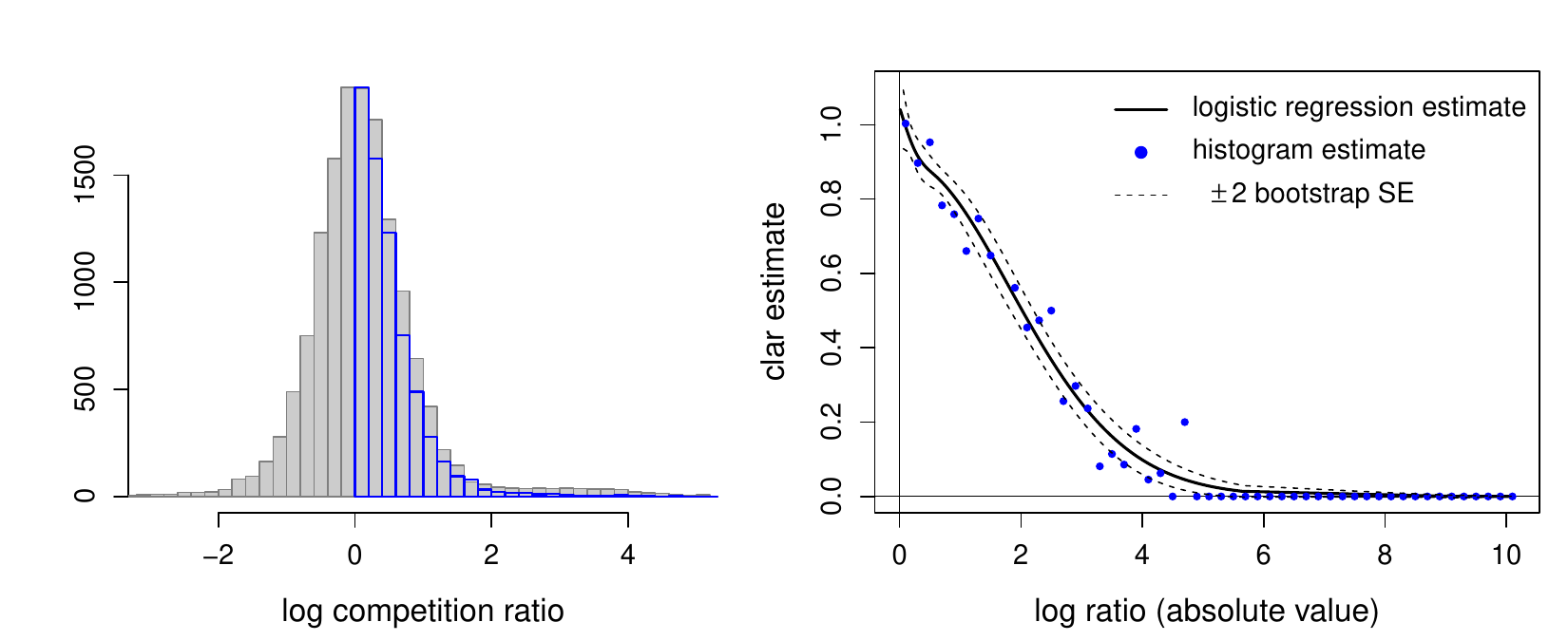}
    \caption{The left panel shows a histogram of $\log_2$ competition ratios from the protein analysis in \cite{byun2023covalent}, where the histogram to the left of zero is reflected to the positive real line. In the right panel, we show two estimates of clar in the protein example: one using a rough histogram estimator (blue), and the other a natural cubic spline logistic regression estimator (black).
    } 
    \label{fig:symmetric-null-lfdr-cysteine}
\end{figure}

Streamlined cysteine activity-based protein profiling (SLC-ABPP) enables proteome-wide detection of amino-acid side chains that can bind small electrophilic compounds, and holds promise for developing new drugs that inhibit specific proteins~\citep{kuljanin2021reimagining}. As one example, \citet{byun2023covalent} profile 3-bromo-4,5-dihydroisoxazole (BDHI)–functionalized compounds to find cysteines they bind. For the $j$-th candidate cysteine site 
($p=14,747$), they measure triplicate controls  $Y_{j1},Y_{j2}, Y_{j3}$ and triplicate BDHI-treated measurements $Z_{j1},Z_{j2},Z_{j3}$. They then compute the competition ratio, defined as 
$$ \mathrm{CR}_j := \frac{\mathrm{median}(Y_{j1},Y_{j2},Y_{j3})}{\mathrm{median}(Z_{j1},Z_{j2},Z_{j3})}.$$
Large values of $\mathrm{CR}_{\idx}$ provide evidence of binding.
After treatment with control (DMSO), respectively BDHI, a cysteine-reactive probe is added that labels only free cysteines. Labeled peptides are then selectively enriched and quantified by mass spectrometry. If BDHI binds a cysteine, that cysteine is no longer available to be bound by the probe, and thus the peptide is under-enriched in the treated sample.  Figure~\ref{fig:symmetric-null-lfdr-cysteine} shows the histogram of $\log \mathrm{CR}_j$ for BDHI-8 as well as the estimated clar. 
Based on the experiment, it is plausible to assume  (i) that $\log_2 \mathrm{CR}_j$ is symmetrically distributed under the null, and (ii) that large values of $\log_2 \mathrm{CR}_j$ provide evidence for the alternative. In analyzing SLC-ABPP experiments, it is common to select cysteine sites with a competition ratio above $4$ which represents $75\%$ attenuation of cysteine probe binding at that site~\citep{kuljanin2021reimagining}. This corresponds to a $\log_2$ ratio of $2$. Using our proposed logistic regression approach, we estimate that roughly 50\% of cysteine sites are false discoveries at this threshold.
The precise estimate is $\widehat{\clar}(2) = 0.506$ with 95\% bootstrap confidence interval $(0.449, 0.563)$.

\subsection{Clar estimates for HIV mutation data}

\begin{figure}[h]
    \centering
    \includegraphics[width=\linewidth]{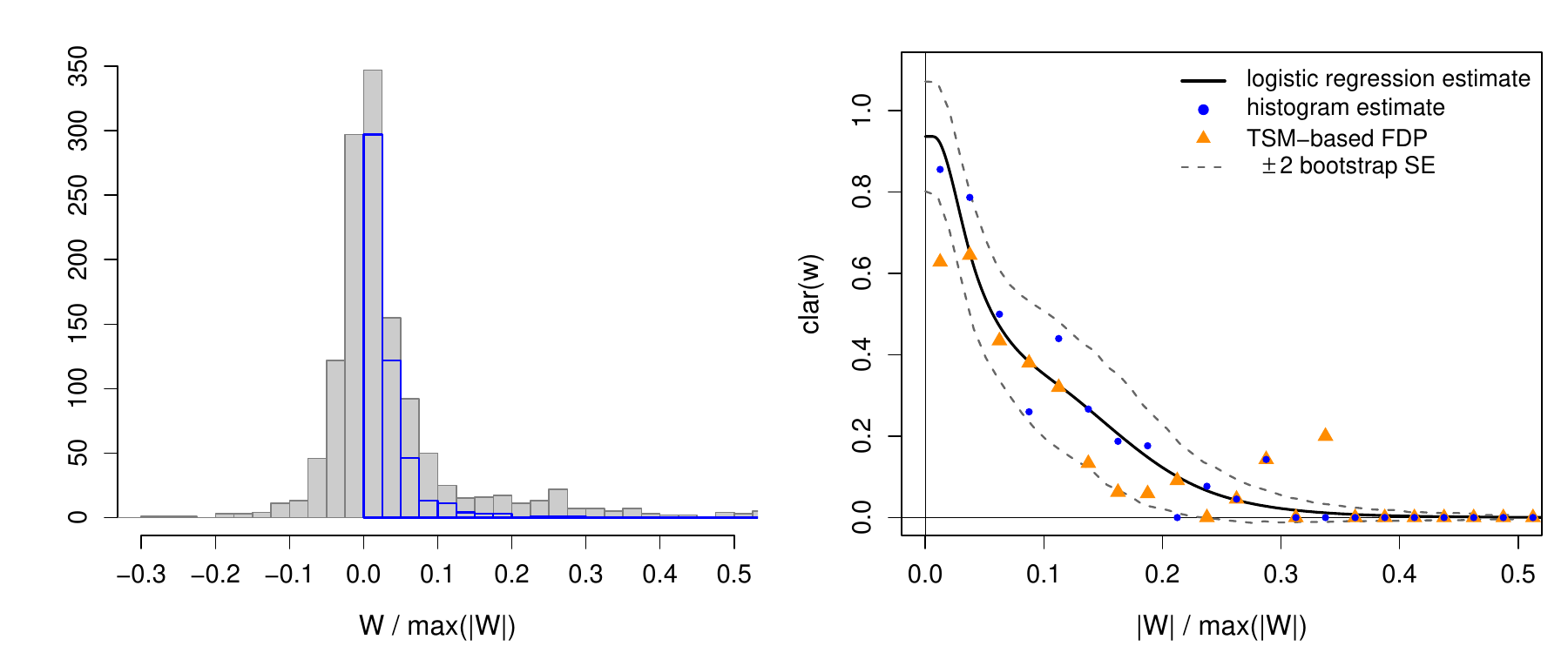}
    \caption{Left: histogram of pooled $W$-statistics after dividing by $\max_{\idx} |W_\idx|$ for each protease inhibitor drug (APV, NFV, ATV, LPV, IDV, SQV, RTV), with the negative tail mirrored onto the positive axis (blue curve). Right: pooled clar estimate across drugs after rescaling by the maximum absolute value $|W|$ within each drug class. Orange triangles show the empirical FDP within bins of roughly $0.025$ units on the normalized scale.} 
    \label{fig:pooled-clar}
\end{figure}

Human Immunodeficiency Virus type 1 (HIV-1) drug resistance poses a challenge to treatment, and understanding the genetic basis of this resistance is important for designing effective new drug therapies \citep{rhee2006genotypic}. 
For each of 16 antiretroviral drugs, the dataset contains log-fold drug resistance measurements for between 630 and 850 HIV-1 samples, and a binary design matrix indicating which mutations are present in each sample. \cite{barber2015controlling} applied the knockoff filter to identify mutations associated with drug resistance, and validated their approach using treatment-selected mutation (TSM) lists: mutations found at significantly higher frequency in patients previously treated with a drug class compared to an untreated control group \citep{rhee2005hiv}.

In this setting, $\clar(w)$ is roughly interpretable as the probability that a mutation is unrelated to drug resistance, given that its $W$-statistic is equal to $w$. 
Our method estimates clar by fitting a smooth curve using logistic regression with a natural cubic spline basis applied to a list of $W$-statistics, computed here using the fixed-X knockoff construction with Lasso-sign-max importance statistic. 
Figure \ref{fig:sign-logistic} illustrates the estimates alongside simple histogram-based estimates for the drug Nelfinavir (NFV).
The estimated clar lies above the TSM-based FDP (orange) in all but one bin for the NFV drug. We plot the fitted curves for six other protease inhibitor drugs from this dataset in Figure \ref{fig:PI-clar-estimates}.
Figure \ref{fig:pooled-clar} shows the clar estimate obtained by pooling $W$-statistics across drugs after dividing each drug's $W$-statistics by their maximum absolute value, placing $W$-statistics from different drugs on the same scale. 
The clar estimate is close to 1 for small $|w|$ and decreases continuously toward zero over the range $|w| \in (0,1/2)$.
It is a conservative proxy for the TSM-based FDP in nearly all of the bins plotted. 
Dashed lines indicate $\pm 2$ bootstrap standard errors based on 200 resamples of the pooled $W$-statistics, refitting the clar curve on each resampled dataset.

\section{Related work and concluding remarks}
\label{sec:discussion}

\paragraph{Modeling the local false discovery rate directly.} \citet{rice2008comment} and~\citet{klaus2011learning} argue that insofar as $\lfdr(w)$ is the central object of a multiple testing analysis, one should directly impose modeling assumptions on $\lfdr(w)$. If the null density $f_0$ is known, then the marginal density $f$ is fully specified by $\lfdr(w)$ via
$$
f(w) =  \left( \int \frac{f_0(w)}{\lfdr(w)} dw   \right)^{-1} \frac{f_0(w)}{\lfdr(w)},
$$
and so, e.g., a parametric model for $\lfdr(w)$ can be estimated via maximum likelihood. \citet{rice2008comment} ask: ``Must one compare the marginal $f$ to an $f_0$ which is assumed to have a specifically Gaussian form, or that of some other parametric family? Might some
advanced form of cross-validation offer a model-free approach?'' In this paper we answer affirmatively in the case wherein $f_0$ is unknown but symmetric. We also directly model $\lfdr(w)$.

We also note that~\citet{efron2001empirical} estimate $\lfdr(\cdot)$ using logistic regression. In their setting, they assume access to $B$ null observations for each $j$, i.e., to  $W_{jk}^0$, $k=1,\ldots,B$ such that $W_{jk}^0$ is equal in distribution to $W_{j}$ for null $j$. Binary classification is then applied to the pairs $(W_j, 1),\, j=1,\ldots,p$ and $(W_{jk}^0,-1),\, j=1,\ldots,p$, $k=1,\ldots,B$, where $\pm 1$ are the labels.

\paragraph{Local false discovery rate estimation under null symmetry.} The idea of modifying knockoff-type procedures to yield an estimate of the local false discovery rate is not new, though prior work has treated it only tangentially. In their discussion,~\citet[Section 6.4]{lei2018adapt} briefly explain how and why one could modify their AdaPT procedure to yield estimates of the local false discovery rate. Meanwhile,~\citet{sesia2020multiresolution} use knockoffs to localize causal variants in genome-wide association studies. In their Supplementary note 5, titled ``Assessing individual significance'',~\citet{sesia2020multiresolution} effectively propose a variant of our clar methodology using binning (i.e., using histogram estimates of the densities). Under two-dimensional null symmetry, \citet{wang2025adaptive} estimate a two-dimensional local false discovery rate using a construction closely related to ours. However, they do not study the 2d-local false discovery rate as an object of interest in its own right; rather, they use it to construct more powerful feature importance statistics for tail FDR control. To our knowledge, the present work provides the first systematic investigation of local false discovery rate estimation under null symmetry.

\paragraph{Frequentist interpretation.} Throughout the paper we have assumed the data are generated from a two-groups model, in which the null statuses $H_1,\dots,H_p$ are identically distributed Bernoulli random variables, and for which the lfdr is defined as the posterior null probability for each null hypothesis. 
In a strictly frequentist model, where the null statuses are fixed and unknown, the lfdr is not meaningful in its original definition, but one can still define a `frequentist' lfdr at $w$ as the conditional probability, given that \textit{some} $W$-statistic is equal to $w$, that the null hypothesis associated with that $W$-statistic is true.
In \cite*{xiang2026frequentist}, this is shown to take the simple expression 
\begin{align*}
    \lfdr_{\text{freq}}(w) := \P(H_J=0 \mid W_J=w \text{ for some }J)=\frac{\sum_{\idx: H_\idx=0} f^{(\idx)}(w)}{\sum_{\idx=1}^p f^{(\idx)}(w)},
\end{align*}
where $f^{(\idx)}$ is the probability density function of $W_\idx$, assumed to be symmetric if $H_\idx=0$.  
From this standpoint, our methodology can be understood as targeting the upper bound under null symmetry, 
\begin{align*}
    \lfdr_{\text{freq}}(w) = \frac{\sum_{\idx:H_\idx=0} f^{(\idx)}(-w)}{\sum_{\idx=1}^p f^{\idx}(w)} \leq \frac{\sum_{\idx=1}^p f^{(\idx)}(-w)}{\sum_{\idx=1}^p f^{(\idx)}(w)} =: \clar_{\text{freq}}(w).
\end{align*}
In this setting, the key regularity condition that enables clar estimation is concentration of the empirical cdf of the $W$-statistics $F_p(w)$ around its expectation, i.e.~the average cdf $\bar{F}(w) = \frac{1}{p}\sum_{\idx=1}^p F^{(\idx)}(w)$.
Under regularity conditions analogous to those in Theorem \ref{thm:asymptotic-bFDR}, consistent estimation of $\clar_{\text{freq}}$ guarantees asymptotic control of FDR at the boundary of the rejection set of the associated thresholding procedure.

\paragraph{Reproducibility:} We provide code to reproduce all figures at the following Github repository: 
\begin{center}
    \texttt{\href{https://github.com/dan-xiang/clar-for-regression}{https://github.com/dan-xiang/clar-for-regression}}
\end{center}

\bibliographystyle{dcu}
\bibliography{reference}

\newpage

\appendix

\section{Natural cubic splines logistic regression}
\label{sec:spline-implementation}

This appendix describes implementation details of the estimator introduced in Section \ref{sec:clar-estimation}. First, some intuition for the method: in practical settings, $\clar(w)$ levels off near 1 as $w$ approaches 0 because the marginal density is typically dominated by the null component near zero. 
We expect a continuous transition of $\clar(w)$ towards zero as $w$ grows, signifying that the alternative component starts becoming more prominent for larger (positive) values of $w$. 
We propose isotonic regression as an initial step for estimating where the clar function first dips appreciably below 1, which informs knot placement for the cubic spline basis:
we place a knot at the point $w_{0.9}$ where this initial isotonic estimate of clar first drops below 0.9, to mark the ``shoulder'' of the change. 
We also place knots to the right of this shoulder, extending out to $w_{0.1}$, the point where the initial estimate first drops below $0.1$. This knot placement scheme is illustrated with dashed and dotted lines in Figure \ref{fig:sign-logistic} (right panel).
Specifically, our method proceeds as follows.

\begin{enumerate}
    \item \textbf{PAVA initialization:} Run isotonic regression\footnote{Isotonic regression operates on a sequence of observations $(x_1,y_1),\dots,(x_n,y_n)$, that minimizes the residual sum of squares subject to a monotonicity constraint, for example, $x_i \leq x_j \Rightarrow \hat{y}^{\text{iso}}_i \geq \hat{y}^{\text{iso}}_j$. This is computed via the pool adjacent violators algorithm (PAVA). Further computational details about the Grenander initialization step can be found in Section~\ref{sec:numerical}.} on the sequence of indicator variables $\mathbf{1}\{W_\idx<0\}$ ordered according to the absolute values $|W|_{(1)}\geq \dots \geq |W|_{(p)}$. Convert the output into an estimate of $\clar$ for each hypothesis $H_\idx$ by computing the odds ratio for the estimated probability $\P(W_\idx<0 \mid |W_\idx|)$. 
    \item \textbf{Internal knot placement:} Put knots at values along the $w$-axis corresponding to where the isotonic estimate of clar first drops below $0.9, 0.7,0.3$, and $0.1$.
    \item \textbf{Boundary knots:} Set the left boundary knot to be $\max\{2w_{0.9}-w_{0.7},\frac{1}{2}w_{0.9}\}$, while enforcing the slope to the left of this value to be zero. 
    Set the right boundary knot to be $2w_{0.1}-w_{0.3}$. 
    \item \textbf{Remove redundant knots:} Delete internal knots until there are at least 30 data points between all adjacent pairs of knots.
    \item \textbf{Fit logistic regression:} Regress $1\{W_\idx<0\}$ onto the modified natural cubic spline basis evaluated at $|W_\idx|$, and set $\widehat{\clar}(w) = \hat{p}(w)/(1-\hat{p}(w))$, where $\hat{p}(w) = \widehat{\P}(W_\idx<0 \mid |W_\idx|=w)$ is the fitted conditional probability.
\end{enumerate}

Next, we explain the choice of boundary knots, which are chosen outside of the estimated transition region $[w_{0.9},w_{0.1}]$ while enforcing linear extrapolation on the right, and zero slope on the left. The left boundary knot is placed to the left of $w_{0.9}$ by half the gap between $w_{0.9}$ and $w_{0.7}$. Then we take a maximum with $\frac{1}{2}w_{0.9}$ to ensure this doesn't fall below zero. Similarly, the right boundary knot is placed to the right of $w_{0.1}$ by half the gap between $w_{0.3}$ and $w_{0.1}$.

The condition $\clar(0)=1$ is implied by $f$ being continuous at zero. A natural way to enforce this is to fit the logistic regression without intercept (see expression \eqref{eq:logistic}). In principle, one could have $\clar(0) > 1$ if the alternative component places more mass to the left of zero than the right near the origin, which would indicate a violation of the positive enrichment assumption. If there are exact zeros, which may occur for Lasso-based $W$-statistics, then we remove these and conduct the analysis conditional on $W_\idx\neq 0$.

\section{More details about interpretations of $\clar(w)$}
\label{sec:suppl_clar}
In this appendix, we fill in the details for the results previewed in the last paragraph of Section~\ref{subsec:surrogate}.
Recall that each observation is distributed according to the marginal density $f$ in~\eqref{eq:marginal} and that we posited that all of $\pi_0$, $f_0$ and $f_1$ are unknown, except that $f_0$ was assumed symmetric. 
However, any triple in
$$ \mathcal{S}(f) := \left\{ (\rho_0, h_0, h_1)\,:\, f(\cdot) = \rho_0 h_0 (\cdot) +(1-\rho_0)h_1(\cdot),\; \rho_0 \in [0,1],\; h_0 \text{ symmetric density},\;h_1 \text{ density} \right\}
$$
by definition gives the same marginal $f$, so 
we cannot distinguish among them based on the data. This also makes $\lfdr(w)$ non-identifiable.
On the other hand, we could define 
$$
\bar{\rho}_0 \equiv \bar{\rho}_0(f) := \sup\left\{ \rho_0\,:\, (\rho_0, h_0, h_1) \in \mathcal{S}(f)   \right\}, 
$$
which, unlike $\pi_0$, is identifiable.
\citet[Theorem 1]{arias2021extending} solve the above variational problem and show that the supremum is attained at $\bar{\rho}_0 = \int f(w) \wedge f(-w) d w$ which 
corresponds to the representation with maximally symmetric null component,\footnote{Below, we do not handle the cases $\bar{\rho}_0=0$ ($f$ is  supported on $[0,\infty)$ or $(-\infty,0]$) and $\bar{\rho}_0=1$ ($f$ is symmetric).}
\begin{equation}
f(w) = \bar{\rho}_0 \bar{h}_0(w) + (1-\bar{\rho}_0) \bar{h}_1(w),\qquad 
\bar{h}_0(w) := \frac{f(w) \wedge f(-w)}{\bar{\rho}_0} ,\  \bar{h}_1(w) := \frac{f(w)-\bar{\rho}_0 \bar{h}_0(w)}{1-\bar{\rho}_0}.
\label{eq:extremal_two_groups}
\end{equation}
Now recall that in the two-groups model in~\eqref{eq:marginal} we used the notation $H_\idx = 0$ or $H_\idx=1$ 
to denote group membership to the null component $f_0$ or the non-null component $f_1$, respectively. 
In analogy, we introduce a random \textit{activity} indicator $\mathcal{A}_\idx$, such that $\mathcal{A}_\idx = 0$ or $\mathcal{A}_\idx = 1$ 
denotes group membership to $g_0$ or $g_1$, respectively, in the specific two-groups decomposition \eqref{eq:extremal_two_groups}. 
It then holds that 
$$
\mathbb P(\mathcal{A}_\idx=0 \mid W_\idx=w) = \frac{\bar{\rho}_0 \bar{h}_0(w)}{f(w)} = \frac{ f(-w)\wedge f(w)}{f(w)} \leq \frac{f(-w)}{f(w)} = \clar(w),
$$
and, moreover, the inequality above is tight as long as $f(-w) \leq f(w)$. 
This is to say that, assuming $f(-w) \leq f(w)$, the clar  {\em coincides} with the $\lfdr$ for the two-groups model in~\eqref{eq:extremal_two_groups} with maximally symmetric null component $\bar{h}_0$.

\begin{remark}
\label{remark:identifiability}
The `activity rate' terminology originates from \cite{mccullagh2018statistical}, who studied the effect of signal sparsity on the marginal density in the signal-plus-Gaussian noise model, and defined the `active' component as the maximally non-$N(0,1)$ part of the marginal. 
In \cite*{xiang2024interpretation}, the complementary local activity rate $(\clar)$ is defined as the complement of the conditional probability that the observation arose from the `active' component, given the value of its test statistic.
By analogy, the `active' component in the current setting is the maximally non-symmetric part of the marginal density, arising in the form of $g_1$ in \eqref{eq:extremal_two_groups}. 
\end{remark}

Above we formally introduced the activity indicator $\mathcal{A}_\idx$. However, one may wonder if $\mathcal{A}_\idx$ can be defined in a way that yields a meaningful probabilistic relationship to the underlying testing problem, i.e., to the non-null indicator $H_\idx$. The following proposition answers this affirmatively: we show that there exists a joint distribution (a coupling) on the triple $(\mathcal{A}_\idx,W_\idx,H_\idx)$ with the correct two-way marginals. Moreover, this joint distribution is such that $H_\idx=0$ always implies $\mathcal{A}_\idx=0$ (that is, null signals are always inactive); but a non-null signal $(H_\idx=1)$ can still be inactive $(\mathcal{A}_\idx=0)$ if its $W$-statistic is small or negative, with probability decreasing in $w$ under the assumption that $f_1(-w) \leq f_1(w)$ for $w>0$ and $f_1(-w)/f_1(w)$ is monotone decreasing.
\begin{proposition}
\label{prop:consistent-triple}
    Suppose $f_0$ is symmetric, and consider the probability distribution on the triple $(\mathcal{A},W,H) \in \{0,1\} \times \R \times \{0,1\}$ defined by the following joint density
\begin{align}
\label{eq:consistent-triple}
    f^*(a,w,h) = \begin{cases}
        \pi_0 f_0(w) \hspace{1em} &\text{if } h=0,a=0 \\
        0 &\text{if } h=0,a=1 \\
        (1-\pi_0) \min\{f_1(-w), f_1(w)\} &\text{if }h=1,a=0 \\
        (1-\pi_0) (f_1(w) - \min\{f_1(-w), f_1(w)\}) &\text{if }h=1,a=1.
    \end{cases} 
\end{align}
Then $f^*$ defines a valid probability density whose two-way marginals for $(H,W)$ and $(\mathcal{A},W)$ are consistent with the two-groups model \eqref{eq:two-group} and the extremal symmetric decomposition \eqref{eq:extremal_two_groups}, respectively. Furthermore, if $f_1(w) \geq f_1(-w)$ for all $w > 0$, then $\clar(w) = \P(\mathcal{A}=0 \mid W=w)$.
\end{proposition}

Since both two-group decompositions in \eqref{eq:marginal} and~\eqref{eq:extremal_two_groups} are compatible with the marginal density $f$, the joint distribution of $(H_\idx,W_\idx)$ cannot be recovered without further assumptions. Some specific possibilities, such as assuming a known or parametric null component, are discussed in Section \ref{sec:discussion}.
The clar sidesteps the issue, depending only on the marginal density $f(w)$ of the observations, and serves as the natural identifiable target of inference.

We close with a final observation about the inequality in~\eqref{eq:clar-lfdr}. Above, we established its sharpness, by which we mean that~\eqref{eq:clar-lfdr} holds with equality at $(\bar{\rho}_0, \bar{h}_0, \bar{h}_1)$, which lies within the class of partially identified two-groups models $\mathcal{S}(f)$. This property and the identifiability of $\clar$ provide strong rationale for estimating it. Nevertheless, one could also ask whether there are any conditions under which the inequality in~\eqref{eq:clar-lfdr} is nearly tight for the true (but unidentified) two-groups triple $(\pi_0, f_0, f_1) \in \mathcal{S}(f)$. As one example, in the right panel of Figure~\ref{fig:twogroups} we see that $\clar(w) = \lfdr(w)$ for all $w \geq 3$. The reason is that $f_1$ is such that $f_1(w) = 0$ for all $w \leq -3$. This observation generalizes in the following way:
the bound $\lfdr(w) \leq \clar(w)$ is asymptotically tight if $f_1(-w)/f_0(w) \to 0$ as $w \to \infty$, since then $\lfdr(w) / \clar(w) \to 1$. To see this, note by symmetry of $f_0$ that the marginal density at $-w<0$ satisfies
\begin{align*}
    f(-w) = \pi_0 f_0(w) + (1-\pi_0) f_1(-w) = \pi_0 f_0(w) (1+o(1)),
\end{align*}
where $o(1)$ is a term tending to 0 as $w \to \infty$. This condition holds, for instance, in one-sided location testing problems where the noise distribution is sub-Gaussian, but fails for heavier-tailed noise distributions such as Laplace or Student-$t$, as detailed in Proposition \ref{prop:verify-tails} below, which we prove in Section \ref{sec:proofs}. 

\begin{proposition}
\label{prop:verify-tails}
    Let $f_1(w) = f_0(w-\mu)$ for some $\mu>0$, where $f_0$ is a symmetric density on $\mathbb{R}$. Consider three cases.
    \begin{enumerate}
        \item Gaussian tails: If $f_0(w) = \frac{1}{\sqrt{2\pi \sigma^2}}\exp\{-w^2/(2\sigma^2)\}$ for some $\sigma>0$, then $\lim_{w \to \infty}\frac{\lfdr(w)}{\clar(w)} = 1$.
        \item Laplace tails: If $f_0(w) = \frac{1}{2b} \exp\{-|w|/b\}$ for some $b>0$, then 
        \begin{align*}
            \lim_{w \to \infty} \frac{\lfdr(w)}{\clar(w)} = \Big(1 + \frac{1-\pi_0}{\pi_0}\; e^{-\mu/b} \Big)^{-1} \geq \pi_0,
        \end{align*}
        with strict inequality if $\pi_0 < 1$.
        \item Student-$t$ tails: If $f_0(w) \propto \Big(1+\frac{w^2}{\nu}\Big)^{-\frac{1+\nu}{2}}$ for some $\nu > 0$, then
        \begin{align*}
            \lim_{w \to \infty} \frac{\lfdr(w)}{\clar(w)} = \pi_0.
        \end{align*}
    \end{enumerate}
\end{proposition}

\section{Proofs of formal results}
\label{sec:proofs}

Before stating Proposition \ref{prop:stone-w}, we record some motivation for the polynomial basis as the featurization in the logistic regression; we don't recommend this choice in practice, but we find the conceptual point still useful.

Recall from \eqref{eq:odds-relation} that the conditional odds $\P(W_\idx < 0 \mid |W_\idx|=w) / (1-\P(W_\idx<0 \mid |W_\idx|=w))$ equals $\clar(w)$. 
A natural way to model these odds is to define a covariate vector $Z_\idx = h(|W_\idx|)$ through some prespecified feature map $h:\R_+ \to \R^J$, and use a logistic regression model with outcomes $1\{W_{\idx}<0\}$ and features $Z_\idx$. 
For example, if we take a polynomial basis,  $h_k(w) = w^k$, the model is 
\begin{align}
\label{eq:logistic}
    \clar(w) = \frac{\P(W_\idx < 0 \mid |W_\idx|=w)}{1-\P(W_\idx < 0 \mid |W_\idx|=w)} = \exp\Big\{ \sum_{k=0}^J \gamma_k w^k \Big\}.
\end{align}
As $J$ grows, this form can approximate any positive continuous function on a compact set, as recorded in Proposition \ref{prop:stone-w} below. In practice, polynomial bases of moderate degree extrapolate poorly beyond the range of the data, which is why our implementation uses a natural cubic spline basis instead (see details in Appendix~\ref{sec:spline-implementation}).

\begin{proposition}
\label{prop:stone-w}
    Suppose $\clar(w)$ is a continuous function of $w$. For any compact set $E \subset \R$ and $\varepsilon > 0$, there exist real numbers $\gamma_1,\dots,\gamma_J$ for which
    \begin{align*}
        \sup_{w \in E} \Big| \; \clar(w) - \exp\Big\{\sum_{k=0}^J \gamma_k w^k \Big\} \; \Big| < \varepsilon.
    \end{align*}
\end{proposition}

\begin{proof}[Proof of Proposition \ref{prop:stone-w}]
    If the odds ratio $\clar(w)$ is continuous in $w$, then $\log \clar(w)$ is as well. It follows from the \href{https://en.wikipedia.org/wiki/Stone%E2%80%93Weierstrass_theorem}{Stone-Weierstrass Theorem} that the log odds can be uniformly approximated on $E$ by a polynomial function $p(w)$. Since $x \mapsto e^x$ is a continuous map, the exponential $\exp\{p(w)\}$ can be made arbitrarily close to $\clar(w)$, uniformly over $E$.
\end{proof}

\paragraph{Proposition \ref{prop:consistent-triple}.}
    Suppose $f_0$ is symmetric, and consider the probability distribution on the triple $(\mathcal{A},W,H) \in \{0,1\} \times \R \times \{0,1\}$ defined by the following joint density
\begin{align}
    f^*(a,w,h) = \begin{cases}
        \pi_0 f_0(w) \hspace{1em} &\text{if } h=0,a=0 \\
        0 &\text{if } h=0,a=1 \\
        (1-\pi_0) \min\{f_1(-w), f_1(w)\} &\text{if }h=1,a=0 \\
        (1-\pi_0) (f_1(w) - \min\{f_1(-w), f_1(w)\}) &\text{if }h=1,a=1.
    \end{cases} 
\end{align}
Then $f^*$ defines a valid probability density whose two-way marginals for $(H,W)$ and $(\mathcal{A},W)$ are consistent with the two-groups model \eqref{eq:two-group} and the symmetric decomposition \eqref{eq:extremal_two_groups}, respectively. Furthermore, if $f_1(w) \geq f_1(-w)$ for any $w > 0$, then $\clar(w) = \P(\mathcal{A}=0 \mid W=w)$.

\begin{proof}[Proof of Proposition \ref{prop:consistent-triple}]
     Note that $f^* \geq 0$, and summing over $a,h \in \{0,1\}$ gives $\pi_0 f_0(w) + (1-\pi_0) f_1(w) = f(w)$, which integrates to 1. It is also evident that the joint distribution of $(H,W)$ is given by the marginals $f^*(0,w,0)+ f^*(1,w,0)=\pi_0 f_0(w)$ and $f^*(0,w,1)+f^*(1,w,1)=(1-\pi_0) f_1(w)$. All that remains is to check that the two-way marginal of \eqref{eq:consistent-triple} for $(\mathcal{A},H)$ matches the symmetric decomposition. For $a=0$, we have
    \begin{align*}
        f^*(0,w,0)+f^*(0,w,1)=\pi_0 f_0(w) + (1-\pi_0)\min\{f_1(-w), f_1(w)\} .
    \end{align*}
    By symmetry of $f_0$, the above is equal to
    \begin{align*}
         \min\{\pi_0 f_0(-w) + (1-\pi_0) f_1(-w),\pi_0 f_0(w) + (1-\pi_0) f_1(w)\} = \min\{f(w) ,f(-w)\} =: \bar{\rho}_0\bar{h}_0(w).
    \end{align*}
    The marginal for $a=1$ follows from the relationship in \eqref{eq:extremal_two_groups}.
    It follows that the two-way marginal for $(\mathcal{A},W)$ in the joint density \eqref{eq:consistent-triple} matches the symmetric two-group model \eqref{eq:extremal_two_groups}. Furthermore, symmetry of $f_0$ and the condition $f_1(w) \geq f_1(-w)$ for $w>0$ imply $f(-w) \leq f(w)$ so that
    \begin{align*}
        \clar(w) = \frac{\min\{f(-w),f(w)\}}{f(w)}=\frac{\bar{\rho}_0 \bar{h}_0(w)}{f(w)} = \P(\mathcal{A}=0\mid W=w)
    \end{align*}
    in the extended model.    
\end{proof}

\paragraph{Proposition \ref{prop:verify-tails}.}
    Let $f_1(w) = f_0(w-\mu)$ for some $\mu>0$, where $f_0$ is a symmetric density on $\mathbb{R}$. Consider three cases.
    \begin{enumerate}
        \item Gaussian tails: If $f_0(w) = \frac{1}{\sqrt{2\pi \sigma^2}}\exp\{-w^2/(2\sigma^2)\}$ for some $\sigma>0$, then $\lim_{w \to \infty}\frac{\lfdr(w)}{\clar(w)} = 1$.
        \item Laplace tails: If $f_0(w) = \frac{1}{2b} \exp\{-|w|/b\}$ for some $b>0$, then 
        \begin{align*}
            \lim_{w \to \infty} \frac{\lfdr(w)}{\clar(w)} = \Big(1 + \frac{1-\pi_0}{\pi_0}\; e^{-\mu/b} \Big)^{-1}
        \end{align*}
        \item Student-$t$ tails: If $f_0(w) \propto \Big(1+\frac{w^2}{\nu}\Big)^{-\frac{1+\nu}{2}}$ for some $\nu > 0$, then
        \begin{align*}
            \lim_{w \to \infty} \frac{\lfdr(w)}{\clar(w)} = \pi_0.
        \end{align*}
    \end{enumerate}

\begin{proof}[Proof of Proposition \ref{prop:verify-tails}]
In the case of Gaussian tails, we have
    \begin{align*}
            \frac{f_1(-w)}{f_0(w)} &= \exp\Big\{ -\frac{1}{2\sigma^2}( (-w-\mu)^2 - w^2) \Big\} \\
            &= \exp\Big\{ -\frac{1}{2\sigma^2}( 2w\mu + \mu^2) \Big\} = \exp\Big\{-\frac{w\mu}{\sigma^2} - \frac{\mu^2}{2\sigma^2} \Big\} \to 0,
        \end{align*}
        as $w \to \infty$.
In the case of Laplace tails, we have
\begin{align*}
    \frac{f_1(-w)}{f_0(w)} = \exp\Big\{-\frac{1}{b} (|-w-\mu|-|w|) \Big\} \to \exp\Big\{-\frac{\mu}{b} \Big\},
\end{align*}
so that
\begin{align*}
    \frac{\lfdr(w)}{\clar(w)} = \frac{\pi_0 f_0(w)}{f(-w)} = \frac{1}{1 + \frac{1-\pi_0}{\pi_0} \frac{f_1(-w)}{f_0(w)}} \to \frac{1}{1+\frac{1-\pi_0}{\pi_0}\exp\Big\{-\frac{\mu}{b} \Big\}},
\end{align*}
as $w \to \infty$. In the case of Student-$t$ tails with $\nu$ degrees of freedom, we have
\begin{align*}
\frac{f_1(-w)}{f_0(w)} = \Big(\frac{1+w^2/\nu}{1+(-w-\mu)^2/\nu}\Big)^{\frac{1+\nu}{2}} \to 1,
\end{align*}
which implies
\begin{align*}
    \frac{\lfdr(w)}{\clar(w)} = \frac{\pi_0 f_0(w)}{f(-w)} = \frac{1}{1+\frac{1-\pi_0}{\pi_0}\frac{f_1(-w)}{f_0(w)}} \to \Big(1+\frac{1-\pi_0}{\pi_0} \Big)^{-1}=\pi_0,
\end{align*}
as $w\to\infty$.
\end{proof}

\begin{lemma}
\label{lem:vandervaart}
    (\cite{van2000asymptotic}, Lemma 2.11) Suppose that $X_n$ converges weakly to $X$ for a random variable $X$ with a continuous distribution function. Then $\sup_{x\in\mathbb{R}}|P(X_n\leq x) - P(X\leq x)| \to 0$.
\end{lemma}

\begin{lemma}
\label{lem:pointwise-uniform}
    Let $F_p(w) = \frac{1}{p}\sum_{\idx=1}^p 1\{W_\idx\leq w\}$ denote the ecdf of $W_1,\dots,W_p$, and suppose $F$ is a continuous cdf on $\mathbb{R}$. If for every $w \in \mathbb{R}$, $F_p(w) \stackrel{\P}{\to} F(w)$, then $\sup_{w \in \mathbb{R}} |F_p(w)-F(w)| \stackrel{\P}{\to} 0$.
\end{lemma}

\begin{proof}
    Let $\{p_k\}$ denote an arbitrary subsequence of $p=1,2,3,\dots$. By the subsequence criterion for convergence in probability (see, e.g.~Theorem 2.3.2 in \cite{durrett2019probability}), it suffices to find a further subsequence along which $\|F_{p_k}-F\|_\infty \to 0$ almost surely.
    
    To this end, let $\mathbb{Q}=\{w_1,w_2,\dots\}$ denote an enumeration of the rational numbers. For each $w \in \mathbb{Q}$, since $F_{p_k}(w)\stackrel{\P}{\to} F(w)$ by assumption, there exists a further subsequence along which $F_{p_k}(w) \to~F(w)$ almost surely. By a standard diagonalization argument, we may construct a subsequence (still denoted $\{p_k\}$ for simplicity) along which $F_{p_k}(w) \to F(w)$ almost surely for every $w \in \mathbb{Q}$. 
Now let $w \in \mathbb{R}$ be any real number. Since rational numbers are dense in $\mathbb{R}$, there exist rational numbers $v,u$ satisfying $v<w<u$ and $F_{p_k}(v) \leq F_{p_k}(w) \leq F_{p_k}(u)$. Letting $k\to\infty$, we have
    \begin{align*}
        F(v) \leq \liminf_{p_k \to \infty} F_{p_k}(w) \leq \limsup_{p_k \to \infty} F_{p_k}(w) \leq F(u),
    \end{align*}
    since convergence was established for the rationals. Letting $v_m \uparrow w$ and $u_m \downarrow w$ where $\{v_m\},\{u_m\}\subset ~\mathbb{Q}$ and repeating the above argument, continuity of $F$ implies the limit of $\{F_{p_k}(w)\}$ exists and is equal to $F(w)$. Since $w \in \mathbb{R}$ is arbitrary, this holds for every $w \in \mathbb{R}$, and Lemma \ref{lem:vandervaart} then implies $\sup_{w \in \mathbb{R}} |F_{p_k}(w)-F(w)| \to 0$ on an event with probability one. 
    We have thus shown that every subsequence admits a further subsequence along which $\|F_{p_k}-F\|_\infty \to 0$ almost surely, so it follows that $\|F_p-F\|_\infty \to 0$ in probability, completing the proof.
\end{proof}

\begin{lemma}
    \label{lem:basic-non-separability}
    Suppose $W$ is drawn from the two-groups model \eqref{eq:two-group} and $X = h(|W|)$ in the setting of Section \ref{sec:supervised-classification}. If $\pi_0 f_0(w)>0$ for all $w$, then $\P[ (2\mathbf{1}\{W_\idx < 0\}-1)\beta^\top X_\idx  \geq 0] <1$ for all $\beta \in \mathbb{R}^J \setminus \{0\}$.
\end{lemma}
\begin{proof}
    By the tower property, it is sufficient to show
    \begin{align*}
     \P[(2\mathbf{1}\{W_\idx<0\}-1) \beta^\top X_\idx < 0 \mid |W_\idx|=w]>0,
    \end{align*}
    for almost all $w>0$.
    The left hand side above is equal to
\begin{align*}
     &=  \mathbf{1}\{\beta^\top h(w) > 0\}\P(W_\idx \geq 0\mid |W_\idx|=w) + \mathbf{1}\{\beta^\top h(w) < 0\}\P(W_\idx < 0\mid |W_\idx|=w) 
\end{align*}
The first term above is positive because $\P(W_\idx \geq 0 \mid |W_\idx|=w)\geq \frac{\pi_0 f_0(w)}{f(-w)+f(w)}>0$; the second term is positive because $\clar(w)>0$, both of which follow from the assumption that $\pi_0f_0(w)>0$ for all~$w$.    
\end{proof}

\begin{remark}
\label{rem:skorokhod}
Assumption \ref{assum:weak_dep} requires uniform convergence in probability of the empirical cdfs among nulls and among all observations to their theoretical counterparts. However, the proof of Theorem \ref{thm:glm-consistency} below assumes almost sure convergence. We justify this as follows. Convergence in probability implies $\max\{\|F_{0p}-F_0\|_\infty,\|F_p-F\|_\infty\} \to 0$ in distribution. By Skorokhod's representation theorem, there exists a vector $(W_1,\dots,W_p,\|F_{0p}-F_0\|_\infty,\|F_p-F\|_\infty)$ with the same distribution as the original but for which the convergence to zero of the last two entries holds almost surely. Since the estimator in Theorem \ref{thm:glm-consistency} is a deterministic function of $(W_1,\dots,W_p)$, consistency towards $\clar(w)$ for a fixed $w$ in the Skorokhod representation implies convergence in distribution to a constant, which implies convergence in probability in the original space.     
\end{remark}

\paragraph{Theorem \ref{thm:glm-consistency}.} 
Suppose Assumptions~\ref{assum:weak_dep}, \ref{assumption:1}, and \ref{assumption:2} hold, 
and that $h_1,\dots,h_J$ are continuous and bounded functions. Then $\hat{\beta}$ exists with probability tending to $1$ and $\hat{\beta} \to \beta^*$ in probability as $p\to \infty$.
\begin{proof}
We first show that $\beta^*$ exists and is unique (a result which is implicit in the theorem statement). To this end, we introduce the notation $Y_{\idx} = \mathbf{1}\{W_\idx<0\}$. Note that we can rewrite,
$$
M^*(\beta) = \E[ -\log\{1+\exp\{-(2Y_\idx - 1)\beta^\top X_\idx\}\}].
$$
Then, writing $\{z\}_- := \max\{-z,0\}$ for the negative part, we have that
$$
M^*(\beta) \leq - \E[ \{(2Y_j-1)\beta^\top X_j\}_-].
$$
On the other hand, for any fixed $b$ with $\|b\|=1$, by Assumption~\ref{assumption:2}, it follows that $\E[  \{(2Y_j-1)b^\top X_j\}_-]>0$. Using continuity of the function $b \mapsto \E[  \{(2Y_j-1)b^\top X_j\}_-]$ and compactness of $\{b:\|b\|=1\}$, we get that $\eta:= \inf_{\|b\|=1}\E[  \{(2Y_j-1)b^\top X_j\}_-] >0$. Now fix $\beta \in \mathbb R^J\setminus \{0\}$ and write $\beta = t \cdot b$ with $t >0$ and $\|b\|=1$. Then,
$$
M^*(\beta) \leq - t \E[  \{(2Y_j-1)b^\top X_j\}_-] \leq -t \eta.
$$
This implies that as $\|\beta\| \to \infty$, we have $M^*(\beta) \to -\infty$, that is, $M^*(\cdot)$ is a coercive function.  Since $M^*(\beta) \in (-\infty,\infty)$ for all $\beta \in \mathbb R^J$ and $M^*$ is continuous, it then follows e.g., by~\citet[Proposition 3.2.1]{bertsekas2009convex}, that the set of maximizers of $M^*$ is nonempty. 

We next show that $M^*$ is strictly concave. For $x=h(|w|)$, write $p_{\beta}(x) = e^{h(|w|)^\top \beta}/(1+e^{h(|w|)^\top \beta})$. Differentiating under the integral gives
\begin{equation}
        \frac{\de^2 M^*(\beta)}{\de \beta^2} =- \E\big[ X_\idx X_\idx^\top p_{\beta}(X_\idx)(1- p_{\beta}(X_\idx))\big] \preceq 0.
\label{eq:hessian}
\end{equation} 
    The condition $\E\big[ X_\idx X_\idx^\top \big] \succ 0$ from Assumption~\ref{assumption:1} implies $\frac{\de^2 M^*(\beta)}{\de \beta^2} \prec 0$; indeed, otherwise there exists some $v \in \R^J\setminus\{0\}$ with 
    \begin{align*}
        \E\big[ v^\top X_\idx X_\idx^\top v \; p_{\beta}(X_\idx) \; (1-p_{\beta}(X_\idx))\big] = 0,
    \end{align*}
    which implies the integrand is zero almost surely. Division on both sides would imply $v^\top X_\idx X_\idx^\top v =~0$ almost surely, a contradiction. Thus $M^*$ is strictly concave. It follows by e.g.,~\citet[Proposition 3.1.1]{bertsekas2009convex}, that the maximizer $\beta^*$ of $M^*$ is unique.

Continuing, by uniqueness of $\beta^*$ and continuity of $M^*(\beta)$, we have the following separation condition:
for every $\varepsilon > 0$
\begin{equation}
        \sup_{\beta \in \R^J : \|\beta-\beta^*\|= \varepsilon} M^*(\beta) < M^*(\beta^*).
\label{eq:separation}
\end{equation}
For the arguments that follow, by Remark \ref{rem:skorokhod}, it suffices to assume that $\|F_p - F\|_\infty \to 0$ almost surely. Fix $\beta \in \R^{J}$ and let $g_{\beta}$ be defined as
\begin{align*}
    g_{\beta}(w) \coloneqq \beta^\top h(|w|) 1\{w<0\}-\log(1+e^{\beta^\top h(|w|)}),
\end{align*}
which is continuous at every $w \neq 0$ and bounded for all $w$. Since $\|F_{p}-F\|_\infty= \sup_{t \in \mathbb{R}} |F_{p}(t) - F(t)| \to 0$, the sequence $\{F_{p}\}$ converges weakly to $F$. 
Since $F$ has a density, it places zero probability on the singleton $\{0\}$. \citet[Theorem 3.10.1(vi)]{durrett2019probability} then implies that
\begin{align*}
    \widehat{M}_{p}(\beta)=\int g_{\beta}(w) dF_p(w) \stackrel{\text{a.s.}}{\to} \int g_{\beta}(w) dF(w)=M^*(\beta).
\end{align*}
On the other hand, one can check that $\beta \mapsto \widehat{M}_{p}(\beta)$ and $\beta \mapsto M^*(\beta)$ are (uniformly) Lipschitz continuous for all $\beta \in \mathbb R^J$. Then, using compactness of $\{\beta:  \|\beta - \beta^*\| \leq \varepsilon\}$, we can strengthen the above convergence to uniform convergence:
$$
\sup_{\beta:  \|\beta - \beta^*\| \leq \varepsilon} \left|  \widehat{M}_{p}(\beta) - M^*(\beta)\right|  \stackrel{\text{a.s.}}{\to} 0.
$$
Using the separation in~\eqref{eq:separation}, the above also implies that,
$$
\mathbb P\left[A_p\right] \to 1 \text{ as } p \to \infty,\, \text{ where }\, A_p:=\left\{\sup_{\beta \in \R^J : \|\beta-\beta^*\|= \varepsilon} \widehat{M}_p(\beta) < \widehat{M}_p(\beta^*)\right\}.
$$
We will now argue that on the event $A_p$, $\hat{\beta}$ exists (that is, $\widehat{M}_p$ has at least one maximizer), and moreover $\|\hat{\beta}-\beta^*\| \leq \varepsilon$ (where this inequality holds for any maximizer $\hat{\beta}$). To show this, it suffices to argue that on $A_p$, we have that,
\begin{equation}\sup_{\beta \in \R^J : \|\beta-\beta^*\|> \varepsilon} \widehat{M}_p(\beta) < \widehat{M}_p(\beta^*).
\label{eq:outer_separation}
\end{equation}
The above implies that we must only search for maximizers in the compact set $\{\beta: \| \beta - \beta^*\| \leq \varepsilon\}$, and the existence therein is guaranteed by continuity of $\widehat{M}_p$ and the extreme value theorem. The fact that $\|\hat{\beta}-\beta^*\| \leq \varepsilon$ then directly follows by the fact that $\hat{\beta} \in \{\beta: \| \beta - \beta^*\| \leq \varepsilon\}$.

Thus, let us show~\eqref{eq:outer_separation} on the event $A_p$. A calculation similar to~\eqref{eq:hessian} yields that $\widehat{M}_p$ is concave (this result can be strengthened to strict concavity with high probability, but we do not need this strengthening here). By contradiction, suppose that $\widehat{M}_p(\beta) \geq \widehat{M}_p(\beta^*)$ for some $\beta$ with $\| \beta - \beta^*\| > \varepsilon$. Then there exists $\lambda \in (0,1)$ such that $\beta_{\lambda} = \lambda \beta + (1-\lambda)\beta^*$ satisfies $\|\beta_{\lambda}-\beta^*\| = \varepsilon$. Moreover, by concavity,
$$
\widehat{M}_p(\beta_{\lambda})  \geq \lambda \widehat{M}_p(\beta) + (1-\lambda) \widehat{M}_p(\beta^*) \geq \widehat{M}_p(\beta^*).
$$
However, the above is a contradiction with the definition of the event $A_p$. Therefore,~\eqref{eq:outer_separation} holds.

\end{proof}

\begin{lemma}
\label{lem:H-continuous}
    Let $G(w) \coloneqq P(|W_\idx| \geq w)$ and $F^{-}(w) \coloneqq P(W_\idx \leq -w)$. If $F(w) = P(W_\idx\leq w)$ has positive density $f$, then $H(t) \coloneqq F^{-}(G^{-1}(t))$ is a monotone, bounded, and continuous function on $(0,1)$, and extends to a uniformly continuous function on $[0,1]$.
\end{lemma}

\begin{proof}
Since the density $f$ exists and is positive, the function $G(w)$ is continuous and strictly decreasing, so $G^{-1}$ is also continuous and strictly decreasing. This implies $H(t)=F^-(G^{-1}(t))$ is a continuous bounded monotone function of $t$. 
Since $H(t)$ is monotone and bounded on $(0,1)$, a continuous extension $\tilde{H}$ on $[0,1]$ is given by $\tilde{H}(1) = \lim_{t \uparrow 1} H(t)$ and $\tilde{H}(0) = \lim_{t \downarrow 0}H(t)$ and $\tilde{H}(t) = H(t)$ for all $t \in (0,1)$. Finally, since $\tilde{H}$ is continuous on the compact set $[0,1]$, it must be uniformly continuous.
\end{proof}

\begin{lemma}
    \label{lem:uniform-convergence-H}
    Suppose $W_1,\dots,W_p \sim F$ with positive density $f$, and Assumption \ref{assum:weak_dep} holds. Define for $w>0$, $t\in (0,1)$,
    \begin{align*}
        F_p^-(w) &\coloneqq F_p(-w), \;\;\;\; G_p(w) \coloneqq F_p(-w)+1-F_p(w) \\
        H_p(t) &\coloneqq F_p^{-}(G_p^{-1}(t)), \;\;\;\; H(t) \coloneqq F^{-}(G^{-1}(t)),
    \end{align*}
    where $G_p^{-1}(t) \coloneqq \sup\{w>0: G_p(w) \leq t\}$, $F^{-}(w)\coloneqq F(-w)$, $G(w) = 1-F(w)+F(-w)$, and $F_p$ is the empirical cdf of $W_1,\dots,W_p$. Then $\|H_p-H\|_\infty \to 0$ in probability as $p\to\infty$.
\end{lemma}

\begin{proof}
Since $H_p$ is a step function and $H$ is monotone, the supremum is achieved at some point $t=k/p$, 
\begin{align*}
    \|H_p-H\|_\infty = \sup_{t \in (0,1)}|H_p(t)-H(t)| = \max_{k=1,\dots,p} |F_p^- (|W|_{(k)}) - F^-(G^{-1}(k/p))|,
\end{align*}
where $|W|_{(1)} \geq \dots \geq |W|_{(p)}$ are the absolute values sorted decreasingly. For any $k=1,\dots,p$, we have
    \begin{align*}
    |F^{-}_p(|W|_{(k)}) - F^{-}(G^{-1}(k/p))| &\leq \|F_p-F\|_\infty + |F^- (|W|_{(k)}) - F^-(G^{-1}(k/p))| \\
    &= \|F_p-F\|_\infty + |H(G(|W|_{(k)})) - H(G_p(|W|_{(k)}))|.
    \end{align*}
    Now fix $\varepsilon>0$. Lemma \ref{lem:H-continuous} implies $H$ is uniformly continuous: there exists $\delta>0$ such that $|x-y| < \delta$ implies $|H(x)-H(y)|<\varepsilon$. Assumption \ref{assum:weak_dep} implies that eventually, $\|G-G_p\|_\infty < \delta$ and thus
    \begin{align*}
        \|H_p-H\|_\infty \leq \|F_p-F\|_\infty + \varepsilon.
    \end{align*}
    Since $\varepsilon$ is arbitrary and $\|F_p-F\|_\infty \stackrel{\P}{\to} 0$, this implies $\|H_p-H\|_\infty \stackrel{\P}{\to} 0$.
\end{proof}

\paragraph{Theorem \ref{thm:consistency-iso}.} 
Suppose $f$ is a positive and continuous probability density function, the true regression function $\mu(w) = \P(W_\idx < 0 \mid |W_\idx|=w)$ is non-increasing in $w$, and Assumption~\ref{assum:weak_dep} holds. Then for each $w$, $\widehat{\clar}_{\textnormal{iso}}(w) \to \clar(w)$ in probability as $p \to \infty$.
\begin{proof}
    Let $H(t) = F^{-}(G^{-1}(t))$ as defined in Lemma \ref{lem:H-continuous}, and let $H_p(t)$ denote its empirical version $H_p(t) \coloneqq F_p^-(G^{-1}_p(t))$, where 
    \begin{align*}
        F_p^-(w) \coloneqq \frac{1}{p}\#\{\idx\leq p:W_\idx \leq -w\}, \;\; G_p(w) \coloneqq \frac{1}{p}\# \{\idx\leq p: |W_\idx| \geq w\},
    \end{align*}
    and $G_p^{-1}(t) \coloneqq \sup\{w>0:G_p(w) \geq t\}$.
    Since $f=F'$ is continuous and positive, we have
    \begin{align*}
        H'(t) = \frac{(F^-)'(G^{-1}(t))}{G'(G^{-1}(t))} = \frac{-f(-G^{-1}(t))}{-f(-G^{-1}(t))-f(G^{-1}(t))} = \mu(G^{-1}(t)),
    \end{align*}
    which implies $H$ is convex by the monotonicity assumption. 
    Letting $\hat{H}_p$ denote the greatest convex minorant of $H_p$, Marshall's inequality then implies
    \begin{align}
    \label{ineq:marshall}
        \| \hat{H}_p - H \|_\infty \leq \| H_p - H \|_\infty \stackrel{\P}{\to} 0,
    \end{align}
    by Lemma \ref{lem:uniform-convergence-H}.
    Since $\hat{\mu}(G^{-1}_p(t))$ is the left-hand slope of the convex function $\hat{H}_p(t)$, we have
    \begin{align*}
        \frac{\hat{H}_p(t)-\hat{H}_p(t-\varepsilon)}{\varepsilon} \leq \hat{\mu}(G^{-1}_p(t)) \leq \frac{\hat{H}_p(t+\varepsilon)-\hat{H}_p(t)}{\varepsilon}.
    \end{align*}
    Pick $t = G_p(w)$ and let $\varepsilon_p \coloneqq \max \left\{\|G_p-G\|_\infty^{1/2},\| H_p - H\|^{1/2}_\infty\right\} \stackrel{\P}{\to} 0$. The first inequality above implies
    \begin{align*}
        \frac{H(G(w))-H(G(w)-\varepsilon_p)}{\varepsilon_p} &\leq \hat{\mu}(w) + \frac{|\hat{H}_p(G_p(w))-\hat{H}_p(G(w))| + |\hat{H}_p(G(w)) - H(G(w))|}{\varepsilon_p} \\
        &+ \frac{|\hat{H}_p(G_p(w)-\varepsilon_p)-\hat{H}_p(G(w)-\varepsilon_p)| + |\hat{H}_p(G(w)-\varepsilon_p) - H(G(w)-\varepsilon_p)|}{\varepsilon_p}. 
    \end{align*}
    Since $\hat{H}_p$ has slopes between 0 and 1, we have $|\hat{H}_p(t_1) - \hat{H}_p(t_2)| \leq |t_1-t_2|$ for any $t_1,t_2$, so the above inequality implies
    \begin{align*}
        \frac{H(G(w))-H(G(w)-\varepsilon_p)}{\varepsilon_p}  \leq \hat{\mu}(w) + \frac{2(\|G_p-G\|_\infty+\|H_p-H\|_\infty)}{\varepsilon_p}
    \end{align*}
    where we have used Marshall's inequality \eqref{ineq:marshall}.
    Similarly, the upper bound implies
    \begin{align*}
        \frac{H(G(w)+\varepsilon_p)-H(G(w))}{\varepsilon_p} \geq \hat{\mu}(w) - \frac{2(\|G_p-G\|_\infty + \|H_p-H\|_\infty)}{\varepsilon_p}
    \end{align*}
    Assumption \ref{assum:weak_dep} implies $\|G_p-G\|_\infty \to 0$ in probability. Therefore,
    \begin{align*}
    \frac{H(G(w))-H(G(w)-\varepsilon_p)}{\varepsilon_p} - \hat{\mu}(w) &\leq \frac{4 \max\left\{\|G_p-G\|_\infty,\|H_p-H\|_\infty \right\}}{\varepsilon_p} = 4\varepsilon_p \stackrel{\P}{\to} 0 \\
    \hat{\mu}(w) - \frac{H(G(w)+\varepsilon_p)-H(G(w))}{\varepsilon_p} &\leq \frac{4 \max\left\{\|G_p-G\|_\infty,\|H_p-H\|_\infty \right\}}{\varepsilon_p} = 4\varepsilon_p \stackrel{\P}{\to} 0.
    \end{align*}
    Since $H'$ exists and is continuous at $G(w)$, the above display implies
    \begin{align*}
        \hat{\mu}(w) \stackrel{\P}{\to} \lim_{\varepsilon \to 0} \frac{H(G(w)+\varepsilon)-H(G(w))}{\varepsilon} = \lim_{\varepsilon \to 0} \frac{H(G(w))-H(G(w)-\varepsilon)}{\varepsilon} = H'(G(w)) = \mu(w)
    \end{align*}
    as $p \to \infty$, from which it follows that $\widehat{\clar}_{\text{iso}}(w) \to \clar(w)$ in probability.
\end{proof}

\paragraph{Theorem~\ref{thm:KDE-consistency}.}
Let $F$ denote the marginal cdf of each of $W_\idx$ and fix $w>0$. Suppose that $F$ has a continuous positive density $f$ in a neighborhood of $w$ and $-w$, and Assumption \ref{assum:weak_dep} holds. Also, assume that the kernel $K$ is supported on $[-1,1]$, is Lipschitz continuous and satisfies $\int_{[-1,1]}K(u) du = 1$. Then there exists a sequence $h_p \to 0$ such that \smash{$\widehat{\clar}_{h_p}(w) \to \clar(w)$} in probability as $p \to \infty$.

\begin{proof}
We can write,
$$
\hat{f}_h(w) = \int \psi_w(u) dF_p(u),\;\; \text{ where }\, \psi_w(u) := \frac{1}{h}K\left(\frac{u-w}{h}\right). 
$$
Write $L$ for the Lipschitz constant of $K$. Then note that $\psi_w(\cdot)$ is supported on $[w-h,w+h]$, has Lipschitz constant $L/h^2$ and is almost everywhere differentiable with derivative $\psi'_w(\cdot)$.

Using integration by parts, we find that:
$$
\begin{aligned}
\left| \int \psi_w(u) dF_p(u) - \int \psi_w(u) dF(u) \right|  &= \left|    \int_{w-h}^{w+h} \psi_w(u)\{ dF_p(u) - dF(u)\}     \right| \\
&= \left|   \int_{w-h}^{w+h} \psi'_w(u)\{F_p(u)-F(u)\} du    \right| \\ 
& \leq 2h \frac{L}{h^2}\|F_p-F\|_\infty = \frac{2L}{h} \|F_p-F\|_\infty.
\end{aligned}
$$
This means that if we choose the bandwidth sequence $h_p$ such that $h_p \to 0$ and $\mathbb E[ \|F_p-F\|_\infty]/ h_p \to 0$, then: 
$$\left| \int \psi_w(u) dF_p(u) - \int \psi_w(u) dF(u) \right| \to 0\; \text{ in probability}.$$
Furthermore:
$$
\begin{aligned}
\int \psi_w(u) dF(u) - f(w) &= \int_{w-h}^{w+h} \frac{1}{h}K\left(\frac{u-w}{h}\right)f(u)du - f(w) \\ 
&= \int_{-1}^1 K(z)\{ f(w + hz) - f(w)\} dz \to 0 \text{ as } h \to 0, 
\end{aligned}
$$
where in the last step we used continuity of $f(\cdot)$ in a neighborhood of $w$ and dominated convergence. Thus we find that:
$$
\hat{f}_{h_p}(w) \to f(w) \text{ in probability as } p \to \infty,
$$
under the conditions we imposed on the sequence $h_p$.
Analogously we can argue that $\hat{f}_{h_p}(-w) \to f(-w)$ in probability, and thus conclude that $\widehat{\clar}_{h_p}(w) \to \clar(w)$ in probability.
\end{proof}

\begin{lemma}
    \label{lem:monotone-clar}
    Let $\alpha \in (0,1)$ and assume (i), (ii), and (iii), defined below: 
    \begin{enumerate}[(i)]
        \item there is a unique $w^*$ for which $\clar(w^*)=\alpha$, 
        \item $\clar$ is continuous at $w^*$ and non-increasing on $[0,w^*+\eta)$ for some $\eta>0$,
        \item there is an open interval $I$ containing $w^*$ such that $\widehat{\clar}$ is uniformly consistent for $\clar$ on~$I$, i.e.~
        \begin{align*}
            \sup_{w \in I} |\widehat{\clar}(w)-\clar(w)| \to 0 \;\; \text{in probability as $p\to\infty$},
        \end{align*}
        and either one of the following conditions holds: either $\widehat{\clar}$ is non-increasing on $[0,w^*)$ with probability tending to 1, or $\widehat{\clar}$ is uniformly consistent for $\clar$ over $[0,w^*)$.
    \end{enumerate}
    Let $\hat{\tau} = \inf\{w\geq 0:\widehat{\clar}(w) \leq \alpha\}$. Then
\begin{align*}
\clar(\hat{\tau}) \to \clar(w^*) = \q \;\; \text{in probability as $p\to\infty$.}
\end{align*}
\end{lemma}

\begin{proof}
Fix $\varepsilon \in (0,\eta)$ small enough so that $[w^*-\varepsilon,w^*+\varepsilon] \subset I$. Since $\clar$ is non-increasing on $(0,w^*+\eta)$ and $w^*$ is unique, we have
\begin{align*}
\clar(w^*-\varepsilon) > \alpha > \clar(w^*+\varepsilon),
\end{align*}
and $\delta \coloneqq \min\big\{\alpha - \clar(w^*+\varepsilon),\clar(w^*-\varepsilon)-\alpha\big\}>0$. By the third assumption, the event
\begin{align*}
    E_p \coloneqq \Big\{ \sup_{w \in I} |\widehat{\clar}(w)-\clar(w)| < \delta\Big\}
\end{align*}
satisfies $\P(E_p) \to 1$ as $p\to\infty$. We will show that on this event, $|\hat{\tau}-w^*|\leq \varepsilon$. To this end, note that on $E_p$, 
\begin{align*}
    \widehat{\clar}(w^*+\varepsilon) - \clar(w^*+\varepsilon) < \delta \leq \alpha - \clar(w^*+\varepsilon),
\end{align*}
which implies $\widehat{\clar}(w^*+\varepsilon) < \alpha$, so that $w^*+\varepsilon \geq \hat{\tau}$ by definition. By similar reasoning,
\begin{align*}
    \clar(w^*-\varepsilon)-\widehat{\clar}(w^*-\varepsilon) < \delta \leq \clar(w^*-\varepsilon)-\alpha,
\end{align*}
which implies $\widehat{\clar}(w^*-\varepsilon) > \alpha$. If $\widehat{\clar}$ is non-increasing on $[0,w^*)$, any $w < w^*-\varepsilon$ also has $\widehat{\clar}(w)>\alpha$, meaning that $[0,w^*-\varepsilon)$ is disjoint from $\{w\geq 0: \widehat{\clar}(w) \leq \alpha\}$. This implies $w^*-\varepsilon \leq \hat{\tau}$. If $\widehat{\clar}$ is uniformly consistent over $[0,w^*)$, then $\sup_{w\in[0,w^*)}|\widehat{\clar}(w)-\clar(w)|<\delta$ with probability tending to 1. On this event, for any $w \leq w^*-\varepsilon$,
\begin{align*}
    \widehat{\clar}(w) > \clar(w) - \delta > \clar(w^*-\varepsilon)-\delta > \alpha,
\end{align*}
from which it follows that $w^*-\varepsilon \leq \hat{\tau}$. 
Thus $|\hat{\tau}-w^*| \leq \varepsilon$ with probability tending to 1 as $p\to \infty$, which implies $\hat{\tau}\to w^*$ in probability since $\varepsilon>0$ can be taken arbitrarily small. Since $\clar$ is continuous at $w^*$, the continuous mapping theorem gives
\begin{align*}
    \clar(\hat{\tau}) \to \clar(w^*) = \alpha
\end{align*}
in probability as $p\to\infty$.
\end{proof}

\paragraph{Theorem \ref{thm:asymptotic-bFDR}.}
Fix $\alpha \in (0,1)$ for which there is a unique $w^*$ satisfying $\clar(w^*)=\alpha$, and suppose that
\begin{enumerate}[(i)]
    \item $\clar$ is continuous at $w^*$ and non-increasing on $[0,w^*+\eta)$ for some $\eta>0$, 
    \item $f,f_0$ are positive and continuous in a neighborhood of $w^*$, and $f_0$ is symmetric about zero, 
    \item $p_0/p\to \pi_0\in(0,1)$ as $p\to\infty$, where $p_0$ is the number of true nulls, 
    \item Assumption \ref{assum:weak_dep} holds, i.e.~$\|F_p-F\|_\infty$ and $\|F_{0p}-F_0\|_\infty$ both tend to zero in probability as $p\to \infty$,
    \item $\widehat{\clar}$ is uniformly consistent for $\clar$ on an open interval containing $w^*$, and one of the following conditions holds: $\widehat{\clar}$ is non-increasing on $[0,w^*)$ with probability tending to $1$, or $\widehat{\clar}$ is uniformly consistent for $\clar$ over $[0,w^*)$.
\end{enumerate}
Then, for the threshold $\hat{\tau}$ defined in \eqref{eq:clar-hat-cutoff}, 
    there exists a sequence of positive numbers $\varepsilon_p \to 0$ such that for any $\delta>0$,
\begin{align*}
\lim_{p\to\infty}\P\big(\FDP([\hat{\tau},\hat{\tau}+\varepsilon_p]) > \alpha + \delta \big) = 0.
\end{align*}

\begin{proof}
Let $\varepsilon_p \coloneqq \max\big\{ \|F_p-F\|^{1/2}_\infty, \|F_{0p}-F_0\|_\infty^{1/2}\big\}$, which tends to zero by Assumption \ref{assum:weak_dep}.
The FDP is expressible in terms of the empirical cdfs (defined in Assumption \ref{assum:weak_dep}): 
    \begin{align*}
        \FDP([\hat{\tau},\hat{\tau}+\varepsilon_p]) 
        = \frac{p_0}{p}\cdot \frac{F_{0p}(\hat{\tau}+\varepsilon_p)-F_{0p}(\hat{\tau}-)}{F_p(\hat{\tau}+\varepsilon_p)-F_p(\hat{\tau}-)}, 
    \end{align*}
    where $G(t-)$ denotes the left limit of a cdf $G$ at $t$. Next, we bound the differences:
    \begin{align*}
        F_{0p}(\hat{\tau}+\varepsilon_p) - F_{0p}(\hat{\tau}-)
        &\leq  F_{0}(\hat{\tau}+\varepsilon_p)-F_{0}(\hat{\tau}-\varepsilon_p') + 2 \|F_{0p}-F_0\|_\infty, \\ 
        F_p(\hat{\tau}+\varepsilon_p) - F_{p}(\hat{\tau}-) &\geq F(\hat{\tau}+\varepsilon_p)-F(\hat{\tau}) - 2 \|F_{p}-F\|_\infty. 
    \end{align*}
    where $\varepsilon_p'/\varepsilon_p \to 0$. By the argument in the proof of Lemma \ref{lem:monotone-clar}, $\hat{\tau}\to w^*$ in probability. Since $F_0$ and $F$ are differentiable and $f_0,f$ are continuous near $w^*$, the mean value theorem and $\varepsilon_p\stackrel{\P}{\to} 0$ imply that
    \begin{align*}
        F_0(\hat{\tau}+\varepsilon_p) - F_0(\hat{\tau}-\varepsilon_p') &=  f_0(\hat{\tau}) \varepsilon_p(1+o_{\P}(1)) \\
        F(\hat{\tau}+\varepsilon_p)-F(\hat{\tau}) &= f(\hat{\tau}) \varepsilon_p(1+o_{\P}(1)),
    \end{align*}
    which are bounded away from zero since $\min\{f_0(w^*),f(w^*)\}>0$.
By construction, $\|F_{0p}-F_0\|_\infty / \varepsilon_p \to 0$ and $\|F_{p}-F\|_\infty / \varepsilon_p \to 0$ in probability. It then follows from the previous displays that
\begin{align*}
    \FDP([\hat{\tau},\hat{\tau}+\varepsilon_p]) \leq \frac{p_0 f_0(\hat{\tau}) \varepsilon_p(1+o_{\P}(1))}{pf(\hat{\tau}) \varepsilon_p(1+o_{\P}(1))} = \lfdr(\hat{\tau})  (1+o_{\P}(1)) .
\end{align*}
    By Lemma \ref{lem:monotone-clar} and symmetry of $f_0$, we have $\lfdr(\hat{\tau}) \leq \clar(\hat{\tau}) \to \alpha$ as $p \to \infty$. Hence 
    \begin{align*}
    \FDP([\hat{\tau},\hat{\tau}+\varepsilon_p]) \leq \clar(\hat{\tau})(1+o_{\P}(1)) \to \alpha,
    \end{align*}
    so that $\P(\FDP([\hat{\tau},\hat{\tau}+\varepsilon_p])>\alpha+\delta) \to 0$ for any $\delta>0$.
 
\end{proof}

\subsection{Diagnostic plots}
\label{sec:diagnostic-plots}

\begin{figure}[h]
    \centering
    \includegraphics[width=\linewidth]{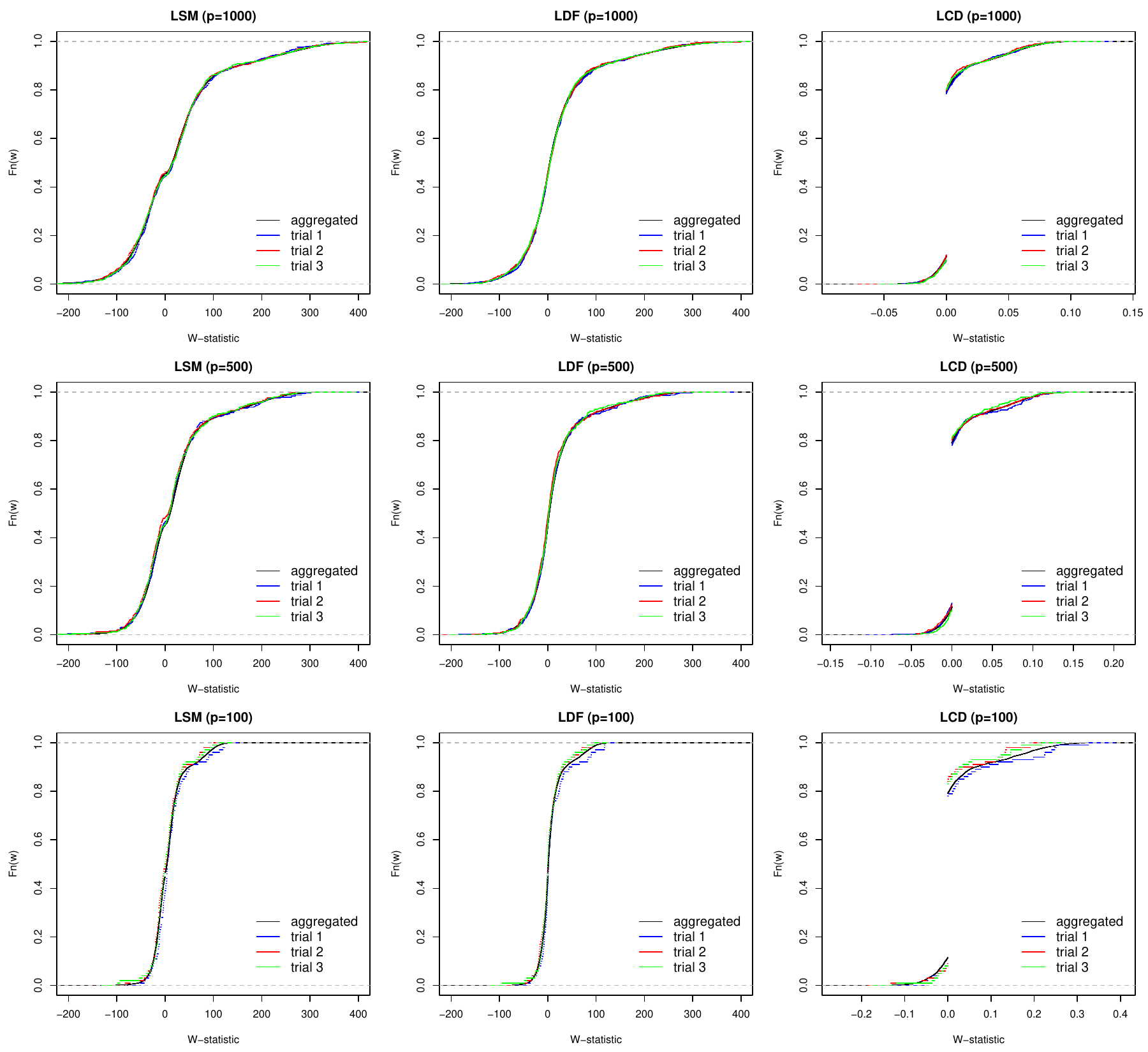}
    \caption{The above plots assess convergence of the empirical cdfs for three types of $W$-statistics where $p=1000$ (top), $p=500$ (middle), and $p=100$ (bottom).}
    \label{fig:ecdf-diagnostic1}
\end{figure}

\begin{figure}[h]
    \centering
    \includegraphics[width=\linewidth]{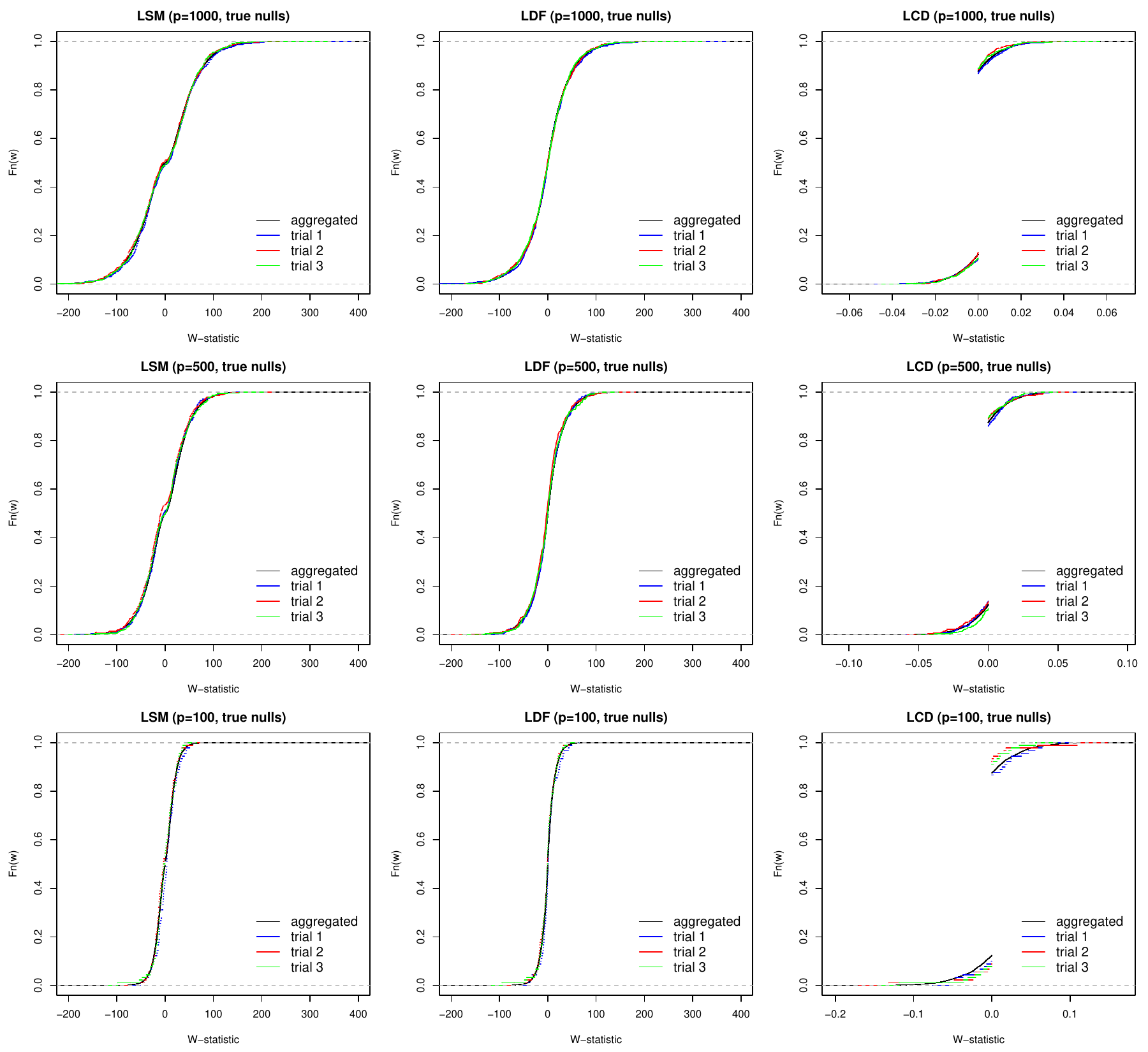}
    \caption{The above plots assess convergence of the empirical cdfs restricted to nulls for three types of $W$-statistics where $p=1000$ (top), $p=500$ (middle), and $p=100$ (bottom).}
    \label{fig:ecdf-diagnostic2}
\end{figure}

Figure \ref{fig:ecdf-diagnostic1} depicts several realizations of the ecdf using three types of importance-statistics, each based on the Lasso regression fit: (i) Lasso-sign-max (LSM), (ii) Lasso $\lambda$ difference (LDF), and (iii) Lasso coefficient difference (LCD); see \cite{knockoff2022} for documentation regarding these and other importance statistics implemented by the R knockoffs package. The diagrams assess convergence of the marginal ecdf for these statistics by plotting three realizations of the ecdf (red, blue, green), overlaid with the ecdf generated from $100 \times p$ many $W$-statistics resulting from pooling across 100 independent simulation runs (black).
The variability of per-run ecdfs decreases as $p$ increases from 100 to 500 to 1000 (bottom to top). Figure \ref{fig:ecdf-diagnostic2} shows the same comparison restricted to the subset of $W$-statistics associated with true nulls.

\begin{figure}[h]
    \centering
    \includegraphics[width=\linewidth]{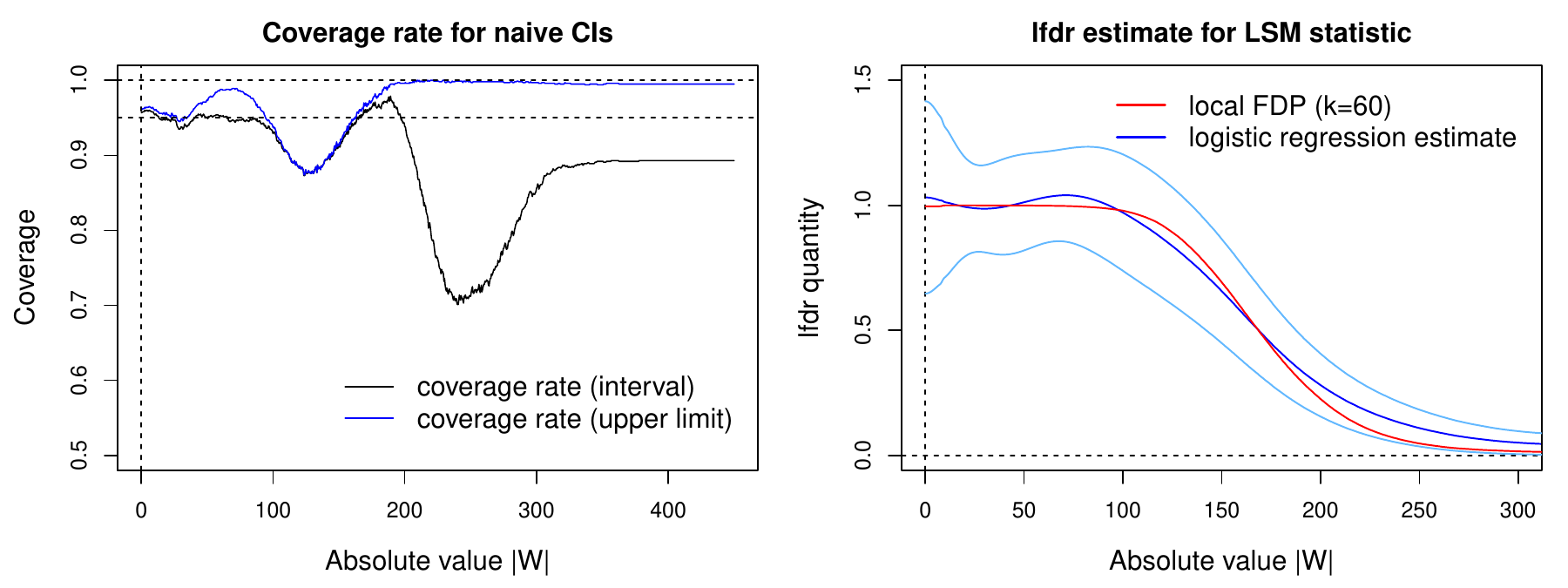}
    \caption{Simulation with delta method CIs. Left: empirical coverage rate over 1000 resamplings of the $y$ vector, holding $\textbf{X}$ and $\beta$ fixed as in the setting of the knockoffs tutorial. Right: the average logistic regression curve (dark blue), with average upper and lower delta-method limits (light blue), and the approximate true lfdr (red). The true lfdr at each $w$ is approximated by the local FDP among the $k=60$ $W$-statistics closest in absolute value to $w$, averaged across the 1000 runs.}
    \label{fig:coverage-naive}
\end{figure}

\begin{figure}[h!]
    \centering
    \includegraphics[width=\linewidth]{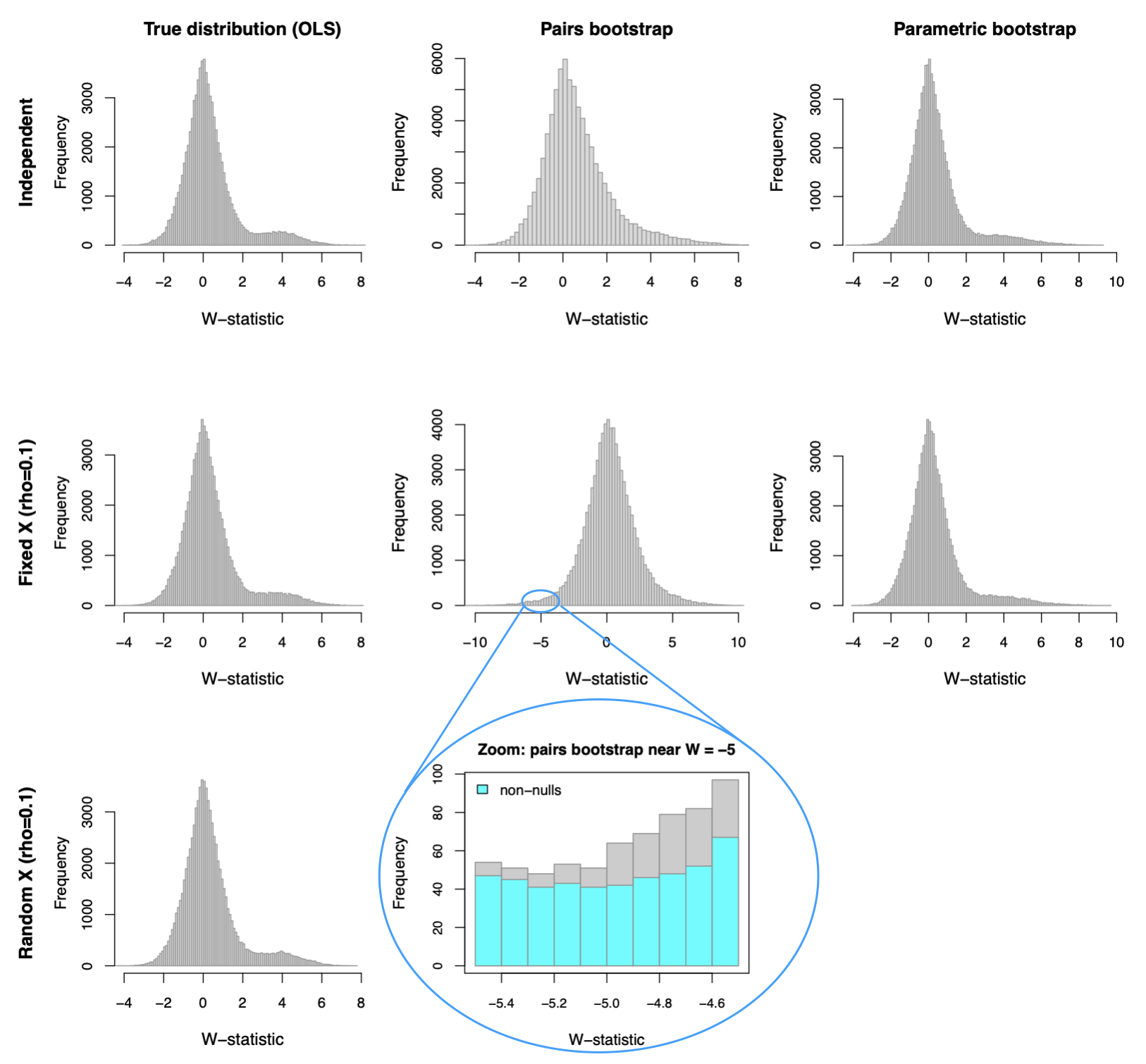}
    \caption{Histograms of the $W$-statistics used to esetimate the clar. For the independent case, the entries of $X$ are i.i.d. $N(0,1)$ random variables. For the correlated case, each row of $X$ is generated from a multivariate normal distribution with toeplitz covariance matrix (1 along the diagonal, $0.1$ adjacent to the diagonal, $0.1^2$ two away from the diagonal, etc.). The bottom row shows the $W$-statistic distribution under a random $X$ design, in which $X$ is redrawn from the same distribution over columns at each iteration, rather than held as fixed across iterations as in the second row. Bootstrap panels are omitted because the knockoffs procedure used to produce $(W_\idx)$ treats $X$ as fixed, so resampling with respect to $X$ would simply produce another realization of the middle row.}
    \label{fig:OLS-statistics}
\end{figure}

\begin{figure}[h!]
    \centering
    \includegraphics[width=0.9\linewidth]{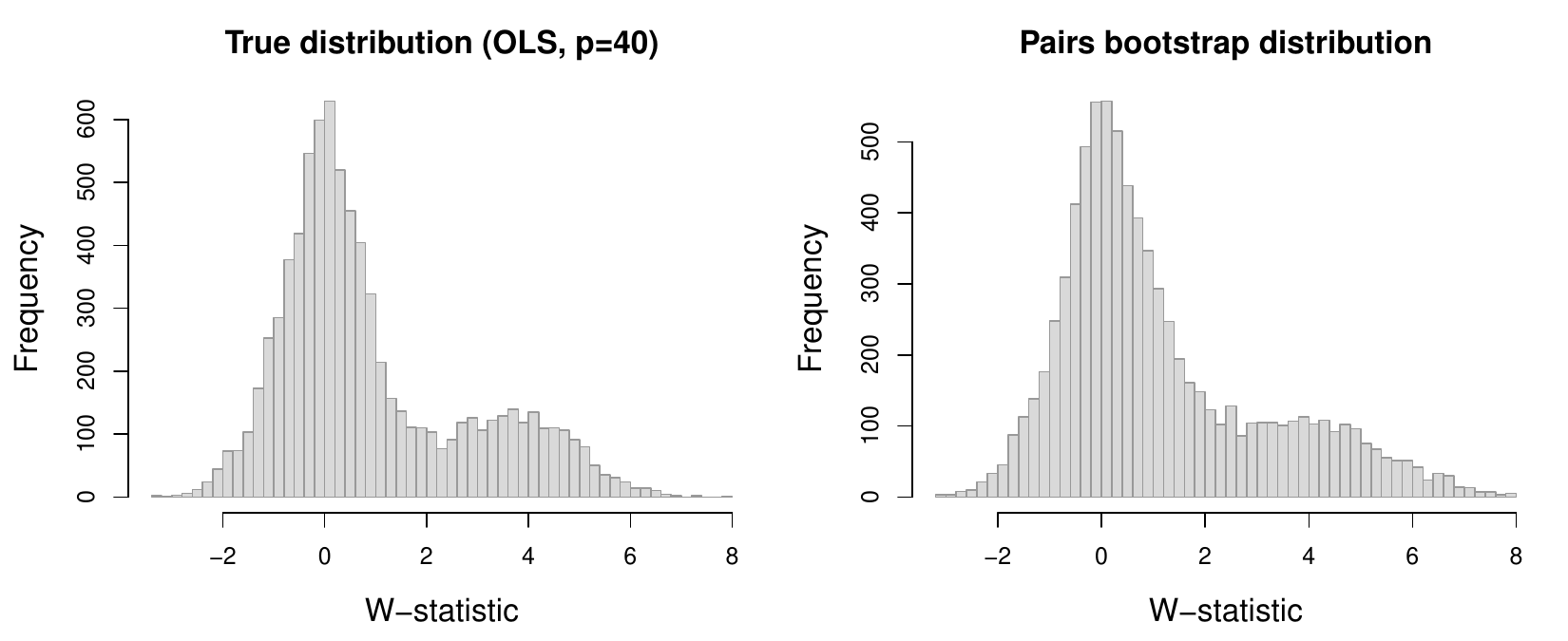}
    \caption{Repeat the experiment in the middle panel of \ref{fig:OLS-statistics} with $n=10^4$, but with $p=40$ and $k=10$ non-nulls, instead of $p=400,k=40$. 
    The histograms now show good agreement, which suggests that the issue with pairs bootstrap in the original example is that the empirical distribution of $(x_1,\dots,x_{400},y)_i$ for $i=1,\dots,10,000$ is not a very good approximation of the true distribution when the number of variables is large.}
    \label{fig:pairs-bootstrap-low-dim}
\end{figure}

\begin{figure}[h]
    \centering
    \includegraphics[width=\linewidth]{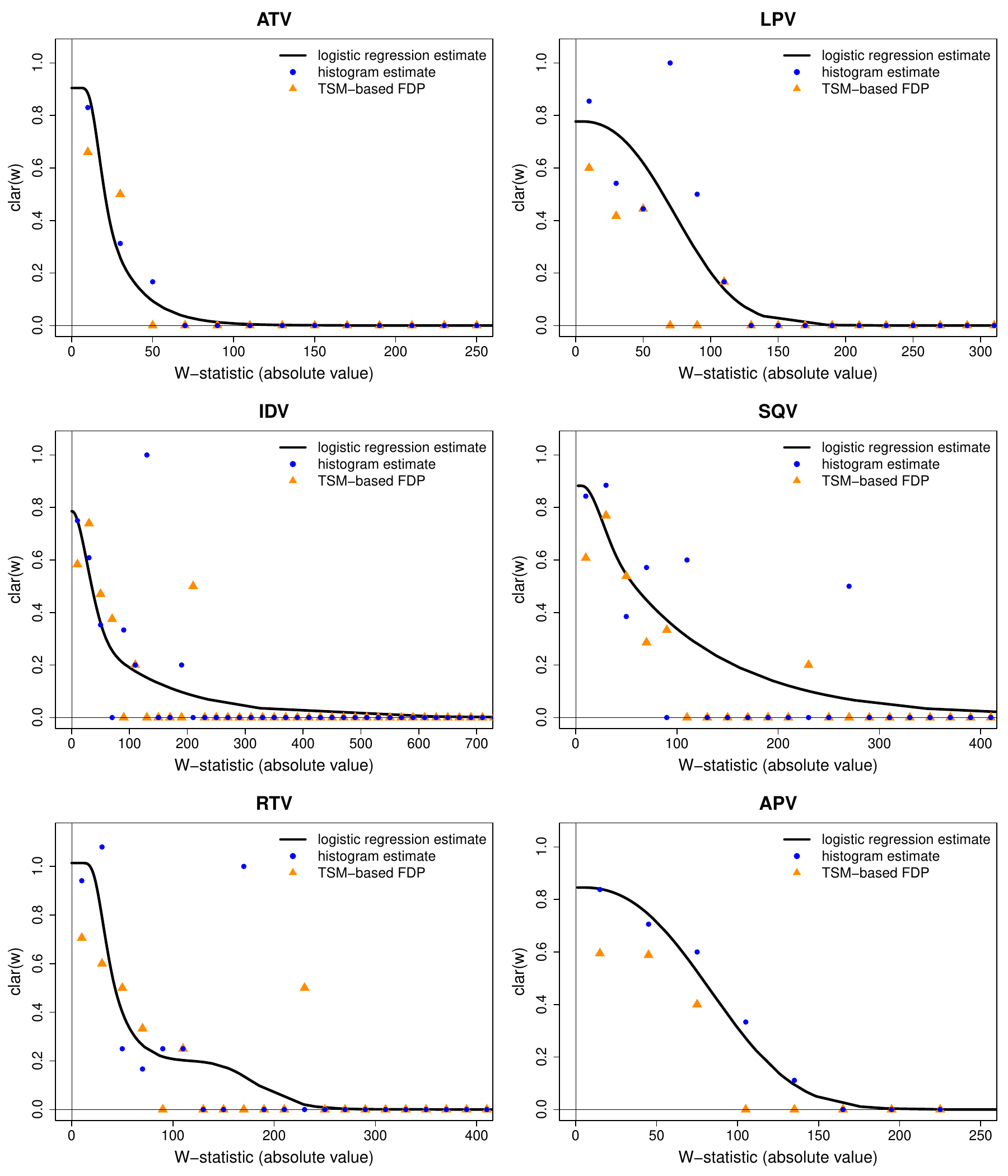}
    \caption{Plotted above are the clar estimates with ground truth labels (obtained from the TSM list) along with the raw histogram estimates overlaid for drugs from the protease inhibitor class in the HIV dataset. The clar estimate for the remaining drug (NFV) was shown in Figure \ref{fig:sign-logistic}.} 
    \label{fig:PI-clar-estimates}
\end{figure}

\end{document}